\documentclass[xelatex, ja=standard, english, letterpaper]{bxjsarticle}
\usepackage[utf8]{inputenc}

\usepackage[T1]{fontenc}
\usepackage{lmodern}
\usepackage{mathpazo}
\usepackage{fullpage}
\usepackage{setspace}
\usepackage{microtype}

\usepackage{amsmath}
\usepackage{amssymb}
\usepackage{amsthm}
\usepackage{mathrsfs}
\usepackage{cases}
\usepackage{mathtools}
\mathtoolsset{showonlyrefs,showmanualtags}
\usepackage{bbm}
\usepackage{empheq}

\usepackage[at]{easylist}
\usepackage{enumitem}
\usepackage{ascmac}
\usepackage{multicol}
\usepackage{appendix}
\usepackage{import}
\usepackage[normalem]{ulem}

\usepackage{xcolor}
\usepackage{graphicx}
\usepackage{float}
\usepackage{caption}
\usepackage{subcaption}
\usepackage{tikz-cd}
\usepackage{listings}

\usepackage[linesnumbered,ruled,vlined]{algorithm2e}

\definecolor{S_pink}{RGB}{255, 64, 159}
\definecolor{V_Orange}{RGB}{204,85,0}
\definecolor{magenta}{RGB}{95,2,31}

\usepackage{natbib}
\usepackage{hyperref}
\hypersetup{
    colorlinks=true,
    linkcolor=blue,
    filecolor=magenta,
    urlcolor=cyan,
}
\newtheorem{thm}{Theorem}
\newtheorem{lem}{Lemma}
\newtheorem{prop}{Proposition}
\newtheorem{cor}{Corollary}

\theoremstyle{definition}

\newtheorem{rem}{Remark}
\newtheorem{df}{Definition}
\newtheorem{eg}{Example}

\newcommand{\mA}{\mathcal{A}}

\newcommand{\mL}{\mathcal{L}}
\newcommand{\mM}{\mathcal{M}}

\newcommand{\E}{\mathbbm{E}}

\DeclareMathOperator*{\argmax}{arg\,max}

\renewcommand{\tilde}{\widetilde}
\renewcommand{\phi}{\varphi}
\renewcommand{\epsilon}{\varepsilon}
\renewcommand{\bar}{\overline}

\title{Robust Learning with Private Information}

\author{
Kyohei Okumura
\footnote{
\href{mailto:kyohei.okumura@gmail.com}{kyohei.okumura@gmail.com}.
}
}
\date{\today}

\begin{document}
\maketitle

\footnotetext{I thank Hersh Chopra, Annie Liang, Harry Pei, Jon Schneider, Marciano Siniscalchi, and Bruno Strulovici for insightful comments.}

\begin{abstract}
    In platform markets, firms increasingly delegate decisions to learning algorithms that must discover from experience how their actions translate into profits. Standard algorithms perform well when the platform's policy is stationary, but against an adaptive policy, their behavior can reveal firms' private information and expose them to surplus extraction. We study repeated trade between a principal and an agent with a fixed private valuation (type). An adaptive principal holds a prior over the agent's type and seeks to maximize expected revenue; the agent commits to a learning algorithm without knowing the principal's prior or whether her policy is stationary or adaptive. An algorithm is \emph{consistent} if it achieves an asymptotically optimal payoff against every stationary policy. We first show that standard consistent algorithms---including all no-external-regret algorithms and canonical stochastic bandit algorithms---allow an adaptive principal to extract the agent's full surplus in the long run. We then construct a consistent, \emph{rate-optimal} algorithm that satisfies \emph{no dynamic rent extraction}: against stationary policies, it achieves the minimax regret rate up to logarithmic factors, and against any adaptive principal, every agent type retains at least the information rent he would earn under the optimal static mechanism for the principal's prior. Finally, these rent guarantees are sharp: no consistent algorithm can raise the guarantee for one type--prior pair without lowering it for another.
\end{abstract}

\newpage

\section{Introduction}
\label{sec:introduction}

Firms increasingly delegate decisions to algorithms. Advertisers automate bidding, retailers automate pricing, and businesses automate the procurement of computing capacity.\footnote{Google's \href{https://support.google.com/google-ads/answer/10966879?hl=en}{bidding documentation} describes how machine learning is used to set bids for individual auctions. Amazon's documentation on \href{https://docs.aws.amazon.com/AWSEC2/latest/UserGuide/ec2-fleet-allocation-strategy.html}{EC2 Fleet allocation strategies} describes how EC2 Fleet selects computing capacity based on price and availability.} An advertiser, for example, may know what an additional ad impression is worth to it but still need to discover how its bids translate into advertising exposure and payments. A bidding algorithm experiments with different bids and uses the observed outcomes to shift toward those with higher payoffs. If the platform running the auctions keeps its policy fixed, standard learning algorithms can learn how to bid optimally.\footnote{\citet{OstrovskySchwarz2023} report that Yahoo raised its reserve price for sponsored search to 10 cents per click in 2003. That rate was
still in place immediately before their 2008 field experiment.}

While the advertiser's algorithm learns, the platform can also learn about the advertiser. As the algorithm adjusts its bids in response to outcomes, its choices can reveal the advertiser's willingness to pay for each ad impression. The platform may use this information to update its policy---for example by setting bidder-specific reserve prices---to extract the advertiser's surplus.
Moreover, the advertiser often cannot tell whether the platform's policy is fixed or adapts to its past bids. Platforms' public documentation often leaves their updating rules unspecified, and observed outcomes may reveal little about whether the policy adapts.\footnote{Analyzing historical Amazon EC2 Spot prices, \citet{AgmonBenYehuda2013} find evidence suggesting that Amazon typically set these prices differently from the way it described in its public documentation at the time.} How much information rent can the advertiser retain while its algorithm learns which bids are profitable?

The contribution of this paper is to show that learning optimally against a stationary principal, whose policy does not respond to the agent's choices, is compatible with protecting the agent's information rents against an adaptive, revenue-maximizing principal---and that this protection comes at no asymptotic cost to learning efficiency. We construct a learning algorithm that guarantees an agent of any type, in the long run, at least what he would earn under the optimal static mechanism for the principal's prior, without knowing the prior or whether the principal adapts. Against a stationary principal, the agent's expected time-average payoff approaches that of the best fixed action at the minimax rate (up to logarithmic factors). Moreover, among learning algorithms that achieve asymptotically optimal payoffs under stationarity, the proposed algorithm is Pareto optimal with respect to its rent guarantees across type--prior pairs.

We study repeated trade between a principal and an agent with a fixed private valuation (type). The agent delegates his decisions to a learning algorithm that chooses among finitely many actions, including an opt-out action, and commits to the strategy it generates.
This commitment reflects automated delegation in practice, where firms deploy pre-configured bidding or pricing software rather than manually overriding decisions period by period.
In each period, the principal sets a \emph{mechanism}---an allocation and a payment for each action---and the agent observes only the outcome of the action he takes. The principal knows the algorithm and observes the agent's past choices. Whereas a stationary principal draws mechanisms independently from a fixed distribution, an adaptive principal can condition mechanisms on the agent's past choices and seeks to maximize expected time-average revenue under a full-support prior over the agent's possible valuations. The agent knows neither the principal's prior nor whether she adapts. We say a learning algorithm is \emph{consistent} if, against every stationary principal, the agent's expected time-average utility converges to the payoff from always choosing the best action, even though he does not know the principal's stationary policy in advance. Consistency is a standard minimal performance requirement for learning algorithms and is our first design goal.

Many standard learning algorithms are consistent but leave the agent vulnerable to rent extraction by an adaptive principal. For example, no-external-regret algorithms are designed to best respond, in the long run, to the opponent's observed play. Against any such algorithm, the principal can extract the full surplus of an agent of any type: she posts a sequence of test prices, infers the agent's valuation from his purchase frequencies, and then charges a price just below the inferred valuation. Canonical stochastic bandit algorithms are similarly vulnerable. Upper confidence bound (UCB) rules fail for a conceptually similar reason: payoff-adaptive exploration reveals the agent's private type through early exploration counts, again enabling full surplus extraction. By contrast, uniform exploration (UE) and successive elimination (SE) fail through a different channel: once they commit to a single action and discard the other actions, the principal can worsen terms on the chosen action. Together, these results show that standard learning rules either reveal the agent's private type during learning or leave him exposed to post-commitment exploitation.

To measure how well an algorithm protects information rents, we use the information rent the agent would earn under static screening as a benchmark. In the static problem, the principal designs a revenue-maximizing mechanism using only her prior.
She can maximize revenue with a posted-price mechanism, and the highest revenue-maximizing price gives each type his lowest payoff among all optimal static mechanisms. We call this payoff his \emph{static screening rent}. We say an algorithm satisfies \emph{no dynamic rent extraction} if, for every full-support prior and every asymptotically revenue-maximizing adaptive response, it guarantees each type at least his static screening rent in the long run.

We construct a consistent algorithm that satisfies no dynamic rent extraction.\footnote{Section~\ref{sec:robust_algorithms} presents two consistent algorithms satisfying no dynamic rent extraction: a simpler rule based on uniform exploration and a rate-optimal rule based on successive elimination. Here we describe the latter.} The idea is to establish a baseline menu before the agent's choices reveal his valuation, and then use the threat of exit to deter departures from that menu. The algorithm first tries each action briefly on a schedule independent of the agent's valuation and records the average allocation and payment for each action, forming an \emph{empirical menu} common to all types. Using successive elimination, it then refines its payoff estimates for the actions to identify the best one, while testing that observed average allocations and payments stay close to the empirical menu. If a test fails, it opts out permanently. These tests limit what the principal can gain by changing terms later, and triggering early exit sacrifices future revenue. Moreover, because every type optimizes against the same empirical menu, the induced outcomes are (approximately) incentive compatible and individually rational under that menu. To maximize asymptotic revenue, the principal must therefore induce an empirical menu close to a static optimum and maintain those terms thereafter, guaranteeing each type at least his static screening rent in the long run. Against a stationary principal, all tests pass with high probability, so the algorithm almost never exits and attains the minimax regret rate up to logarithmic factors.

We also show that the proposed algorithm's rent guarantees are sharp across types and priors: no consistent algorithm can improve them at one type--prior pair without lowering them at another. Any algorithm whose guarantees weakly dominate ours must itself satisfy no dynamic rent extraction, so it suffices to show that every consistent algorithm satisfying no dynamic rent extraction guarantees each type exactly his static screening rent. The argument highlights the distinct roles of prior-robustness and consistency. First, prior-robustness---protecting rents simultaneously across all priors without knowing the principal's prior---caps the principal's asymptotic revenue at the optimal static revenue. Second, consistency lets the principal attain this cap by posting a price slightly below the highest revenue-maximizing static price and holding it fixed throughout the interaction. As the horizon grows, she shrinks the discount to zero, and her revenue converges to the cap. In the limit, each type earns exactly his static screening rent.

In many applications, the agent must also learn his own valuation from experience. We consider an extension in which each period's valuation is an independent draw from a distribution with a fixed mean that the agent does not know; the agent privately observes the realized value only when he receives the good. Our proposed algorithm and main results extend to this setting.

\paragraph{Related literature}

Our problem is one of dynamic screening: information revealed by the agent's current choices can affect the terms he faces later, giving an informed agent an incentive to conceal his type \citep{HartTirole1988,LaffontTirole1988,Devanur2019,Amin2013-strategic}. In this literature, the agent knows the principal's strategy and best responds to it. Our agent, by contrast, must learn which actions are profitable without knowing the principal's policy or her prior over his valuation. We ask what payoff the agent can guarantee by committing to a learning rule.\footnote{\citet{Shalev1994-zn} characterizes equilibrium payoffs in repeated games with incomplete information and known-own payoffs, which is the same payoff structure as in our model. Our guarantees come from commitment to a learning rule rather than from equilibrium continuation play.}

Our learning criteria draw on the literature on learning in games. Beyond explaining how equilibrium can arise from adaptive play \citep{Foster1997-fm,Hart2000-rm}, this literature asks what payoff a learning rule can guarantee an individual player \citep{fudenberg1995consistency,Fudenberg1999-eq,Hart2001-ey}. We share its aim of learning well against stationary opponents while retaining a guarantee against adaptive ones. Universal consistency \citep{fudenberg1995consistency} requires a learning rule to do as well as every fixed action against the opponent's realized play; it implies consistency and guarantees a player at least his maximin payoff against any opponent.\footnote{Universal consistency is a version of no external regret \citep{auer2002nonstochastic}. It implies our no-external-regret criterion (Definition~\ref{df:no_ER}), and thus universally consistent learners are subject to full surplus extraction in our setting.} However, as \citet{Fudenberg1999-eq} note, responding well to an opponent's realized play is not the same as influencing the opponent's future play.\footnote{Responding well to observed play can be optimal when an individual's actions have little influence on the future play he faces, as under anonymous random matching.}
In our model, by responding well to the observed terms, the agent reveals his persistent valuation, which allows the principal to extract the full surplus. In contrast, our learning rules are consistent and yet guarantee every type at least his static screening rent.

Recent work on learning against optimizing opponents \citep{Braverman2017-tm} studies the \emph{strategic robustness} of learning algorithms. This literature compares the optimizer's average payoff in the repeated interaction with her \emph{Stackelberg value}: her payoff when she commits to a fixed stage-game strategy and the learner best responds. A learning algorithm is said to be strategically robust if it caps the optimizer's average payoff at the Stackelberg value; such algorithms have been constructed in a range of settings \citep{Deng2019-ox,Mansour2022-kp,Arunachaleswaran2025_swap,Kumar2024-fh,Zhao2026-strategic}. However, when the learner's valuation is fixed, as in our model, the Stackelberg value for each type equals that type's full surplus, so strategic robustness offers no rent protection.\footnote{\citet{Brown2023-pc} provide a general best-response-query method that yields full surplus extraction in our setting against anytime no-external-regret learners. Our impossibility result (Theorem~\ref{thm:impossibility_noER}) requires only vanishing regret at the terminal horizon; see the discussion following that theorem for the precise comparison.} Our robustness criterion, \emph{no dynamic rent extraction}, is more demanding: it requires each type to retain at least his static screening rent against any asymptotically revenue-maximizing principal with any full-support prior, and hence guarantees strictly positive rents to some types under some priors.\footnote{\citet{Arunachaleswaran2024-df} study learning algorithms that are Pareto optimal with respect to the learner's payoff across the optimizer's possible payoff functions, while \citet{Arunachaleswaran2025-unknown} study optimal algorithmic commitments under a known prior over those functions. In both settings, the optimizer knows the learner's utility function. By contrast, in our model, the principal does not know the agent's valuation ex ante.}

\paragraph{Roadmap}
Section~\ref{sec:model} presents the environment and the consistency criterion. Section~\ref{sec:limits_standard_algorithms} shows the limits of standard algorithms. Section~\ref{sec:protecting_information_rents} introduces the static screening benchmark and the no-dynamic-rent-extraction criterion. Sections~\ref{sec:robust_algorithms} and~\ref{sec:optimality} construct rent-protecting algorithms and establish their Pareto optimality with respect to rent guarantees. Section~\ref{sec:extensions} studies the case in which the agent also learns his valuation. Section~\ref{sec:conclusion} concludes. All proofs are in the appendix.

\section{Model}
\label{sec:model}

We study repeated trade between an agent with a fixed private valuation (a buyer) and a principal (a seller) who controls how the agent's actions translate into allocations and payments. The agent knows his valuation and observes only the realized outcomes of his own actions. He does not know whether the principal uses a fixed policy or adapts her policy to his past behavior. He therefore delegates his decisions to a learning algorithm and commits to the strategy it generates for the $T$-period interaction. We first describe the repeated interaction and learning algorithms, then define stationary and adaptive principals and the asymptotically optimal responses of an adaptive principal. Finally, we introduce consistency, the baseline performance requirement for learning algorithms.

\subsection{Repeated interaction and learning algorithms}
\label{sec:learner}

Two players interact for $T\in\mathbb{Z}_{>0}$ periods. The agent's stage-game action set is finite and denoted by $\mathcal{A}$. In each period, the principal chooses a \emph{mechanism}. A mechanism is a pair $(x,p)$ where $x:\mathcal{A}\to\{0,1\}$ is an allocation rule and $p:\mathcal{A}\to[0,p_{\mathrm{max}}]$ is a payment rule for some $p_{\mathrm{max}} >0$.\
Let
\[
\mathcal{M}\coloneqq \left\{ (x,p) \in \{0,1\}^{\mathcal{A}}\times [0,p_{\mathrm{max}}]^{\mathcal{A}} \colon x(a_0)=0, \ p(a_0) = 0 \right\}
\]
denote the set of mechanisms, where $a_0 \in \mA$ denotes the opt-out action. We assume that $\mathcal{A}$ contains the opt-out action and at least one other action.

The agent has a fixed private type $\theta_i \in \Theta$, where $\Theta = \{\theta_1, \theta_2, \dots, \theta_n\} \subseteq [0,p_{\mathrm{max}})$ is the finite set of possible valuations with at least two strictly positive elements. We index the types in increasing order, so that $\theta_1<\cdots<\theta_n$. In each period $t \in \{1,\dots,T\}$, the agent chooses $a_t\in\mathcal{A}$ and the principal chooses $(x_t,p_t)\in\mathcal{M}$ simultaneously. The stage-game payoff of an agent of type $\theta_i$ is
\[
u(a_t,(x_t,p_t),\theta_i)\coloneqq \theta_i\,x_t(a_t)-p_t(a_t),
\]
and the principal's stage-game payoff is revenue $p_t(a_t)$. Both players evaluate outcomes by undiscounted time-average payoffs.

At the end of period $t$, the agent observes $(a_t,x_t(a_t),p_t(a_t))$ but not $(x_t(a),p_t(a))$ for actions $a\neq a_t$.\footnote{In the online learning literature, this type of feedback is called \emph{bandit feedback}. Note that the feedback here is slightly richer than standard bandit feedback, as the agent observes not only the reward but also the allocation and payment. The setup is closely related to repeated games with imperfect monitoring \citep{Lehrer2016-rf} and repeated games with incomplete information and private learning \citep{Wiseman2012-zq}.} The principal observes the agent's action $a_t$ and the mechanism $(x_t,p_t)$ she chose, but does not observe $\theta_i$. Their histories before period $t$ are
\[
h_t^A \coloneqq \bigl( a_s,x_s(a_s),p_s(a_s) \bigr)_{s=1}^{t-1}, \quad h_t^P \coloneqq (a_s,x_s,p_s)_{s=1}^{t-1},
\]
where $h_t^A$ and $h_t^P$ are the agent's and principal's histories, respectively. Their behavioral strategies map their respective histories to mixed strategies for the stage game. The action and mechanism draws in a period are independent conditional on the past.

Before the interaction begins, the agent privately supplies his valuation $\theta_i$, along with horizon $T$, to a learning algorithm and commits to the strategy it generates.

\begin{df}[Online learning algorithm]
  An \emph{online learning algorithm} $\mathcal L$ maps a pair $(\theta_i,T)$ into a behavioral strategy $\mathcal L(\theta_i,T)$ for the agent in the $T$-period game.
\end{df}

\begin{rem}
  In the baseline model, the agent knows his valuation before the interaction begins. Section~\ref{sec:extensions} considers a setting in which an agent, whose period-valuations are i.i.d.-drawn from a fixed but unknown distribution, learns about his mean valuation while learning how best to respond to the principal's policy. The main results carry over to this setting.
\end{rem}

\subsection{Stationary and adaptive principals}
\label{sec:adaptive_env_def}

The principal knows $\mL$ but does not observe the agent's type. She has a full-support prior $\pi\in\Delta(\Theta)$ and seeks to maximize expected time-average revenue over her feasible strategies. The prior is unknown to the agent and the learning algorithm. For example, a platform may form its beliefs about users' types using proprietary data.

We distinguish two types of principals by the strategies available to them. A \emph{stationary principal} is restricted to drawing $(x_t,p_t)\sim\mu$ independently across periods from a distribution $\mu\in\Delta(\mM)$ that remains fixed throughout the interaction. An \emph{adaptive principal} may use any behavioral strategy $\sigma_P^T$ that maps her observed history to a distribution over mechanisms. The agent does not know which type of principal he faces.

The horizon $T$ and the online learning algorithm $\mathcal{L}$ are exogenously given. The timing is as follows:
\begin{enumerate}
    \item The agent privately observes his valuation $\theta_i$.
    \item The online learning algorithm $\mathcal{L}$ generates the agent's strategy $\mathcal{L}(\theta_i, T)$, and the agent commits to it.
    \item The principal chooses her strategy based on her belief $\pi$, knowing $\mL$ and $T$.
    \item The interaction then begins and lasts for $T$ periods. 
\end{enumerate}

\subsection{Payoffs and approximately optimal principal responses}

Given the agent's valuation $\theta_i$, the learning algorithm $\mL$, and the principal's strategy $\sigma_P^T$, define the expected time-average allocation and payment for each type by
\begin{align}
 \bar x_{i,T}\coloneqq \mathop{\mathbb E}_{\mathcal L(\theta_i,T),\sigma_P^T}
                    \left[\frac1T\sum_{t=1}^Tx_t(a_t)\right],\quad
 \bar p_{i,T}\coloneqq \mathop{\mathbb E}_{\mathcal L(\theta_i,T),\sigma_P^T}
                    \left[\frac1T\sum_{t=1}^Tp_t(a_t)\right].
\end{align}
The expectation is taken over the sequence of actions and mechanisms $(a_t, x_t, p_t)_{t=1}^T$ generated by $\mathcal L(\theta_i,T)$ and $\sigma_P^T$. We suppress the dependence on these strategies when it is clear from context. A type-$\theta_i$ agent's expected time-average payoff is $\bar u_{i,T} \coloneqq \theta_i\bar x_{i,T}-\bar p_{i,T}$.

Writing $\pi_i \coloneqq \pi(\theta_i)$, define the principal's expected revenue under prior $\pi$ given the learning algorithm $\mL$ and the principal's strategy $\sigma_P^T$ by
\begin{align}
\operatorname{Rev}_T(\sigma_P^T;\pi,\mathcal L)
\coloneqq \sum_i\pi_i\bar p_{i,T}.
\label{eq:adaptive_revenue_temp}
\end{align}

As the agent's strategy may respond discontinuously to feedback, an adaptive principal may not have an exact best response.
We therefore assume that the adaptive principal chooses an asymptotically optimal response. To define such responses, first write the principal's value as
\begin{align}
V_T(\pi,\mathcal L)
\coloneqq \sup_{\sigma_P^T\in\Sigma_P^T} \operatorname{Rev}_T(\sigma_P^T;\pi,\mathcal L),
\label{eq:adaptive_value_temp}
\end{align}
where $\Sigma_P^T$ denotes the set of strategies available to the adaptive principal.

\begin{df}[Asymptotically optimal response]
\label{df:adaptive_principal_temp}
We say a sequence of strategies $(\sigma_P^T)_T$ for the principal is \emph{asymptotically optimal} given $\pi$ and $\mathcal L$ if
\begin{equation}
V_T(\pi,\mathcal L)
-\operatorname{Rev}_T(\sigma_P^T;\pi,\mathcal L)\longrightarrow 0 \quad (T \to \infty).
\label{eq:asymptotically_optimal_response_temp}
\end{equation}
Denote by $\operatorname{AO}(\pi,\mathcal L)$ the set of all asymptotically optimal sequences of strategies for the principal.
\end{df}
Note that $\operatorname{AO}(\pi,\mathcal L)$ is nonempty since the definition of the supremum guarantees a feasible strategy with revenue at least $V_T(\pi,\mathcal L)-1/T$ for every $T$.

\subsection{Consistency}
\label{sec:AOS}

Our first design requirement is that the algorithm be \emph{consistent}. If the agent knew the principal's stationary strategy, he could maximize expected utility by choosing the same best-response action in every period.
Consistency requires the algorithm's expected time-average utility to converge to this benchmark.

To formalize this requirement, fix a horizon $T$ and a type $\theta_i\in\Theta$.
Given a stationary strategy $\mu\in\Delta(\mathcal M)$ of the principal, the \emph{external regret} of learning algorithm $\mathcal L$ is the difference between the optimal per-period expected utility under $\mu$ and the agent's average expected utility:\footnote{The cumulative quantity $T \times \operatorname{ER}(\mathcal L(\theta_i,T); T, \theta_i, \mu)$ is also called \emph{pseudo-regret} in the online learning literature.}

\begin{align}
\operatorname{ER}(\mathcal L(\theta_i,T); T, \theta_i, \mu)
\coloneqq \max_{a\in\mathcal A} \mathop{\mathbb E}_{(x,p)\sim\mu}
\left[u(a, (x,p),\theta_i)\right] - \mathop{\mathbb E}_{\mathcal L(\theta_i,T),\mu}
\left[\frac{1}{T}\sum_{t=1}^T u(a_t, (x_t,p_t),\theta_i)\right].
\label{eq:wer_definition}
\end{align}
Note that external regret is always nonnegative when the principal's policy is stationary.

\begin{df}[Consistency]
  \label{df:consistency}
  An online learning algorithm $\mathcal L$ is \emph{consistent} if its average expected utility converges to the optimal per-period expected utility uniformly over agent types and stationary policies:
  \begin{align}
  \label{eq:consistency_definition}
  \lim_{T\to\infty}
  \sup_{\theta_i\in\Theta, \ \mu\in\Delta(\mathcal M)}
  \operatorname{ER}(\mathcal L(\theta_i,T); T, \theta_i, \mu)=0.
  \end{align}
  \end{df}

Consistency is often imposed as a minimal performance requirement for learning algorithms.\footnote{For example, \citet{fudenberg1995consistency} propose as a desideratum that a learning rule be \emph{consistent}, that is, ``do at least as well as playing the best response to the empirical average of play if the opponents' play is given by independent draws from a fixed distribution,'' and refer to consistency as ``the minimal learning property.''} It is satisfied by canonical algorithms for bandit problems, including UCB, successive elimination, and EXP3 \citep{Lattimore2020-xh}, and by familiar learning rules in economics, including fictitious play and regret matching \citep{fudenberg1998theory}. It also holds for tabular Q-learning with UCB exploration \citep{NEURIPS2018_d3b1fb02}.

To discuss how quickly the gap between the optimal expected utility and the agent's expected utility vanishes, we define \emph{minimax external regret} by
\begin{align}
  \inf_{\mathcal L} \sup_{\theta_i\in\Theta, \ \mu\in\Delta(\mathcal M)}
  \operatorname{ER}(\mathcal L(\theta_i,T); T, \theta_i, \mu).
\end{align}
The infimum ranges over all online learning algorithms. Thus, the minimax external regret is the smallest worst-case regret achievable by any algorithm. For the fixed action and type sets considered here, the standard bandit minimax rate gives the minimax external regret of $\Theta(T^{-1/2})$ \citep{Lattimore2020-xh}.\footnote{The upper bound follows from EXP3 after rescaling utilities. For the lower bound, fix a positive type $\theta$ and compare two stationary policies under which an action $b\neq a_0$ charges $\theta/2$ and allocates with probabilities $1/2+\varepsilon$ and $1/2-\varepsilon$, respectively; all other actions give $(0,0)$. Since payments are constant, allocation--payment feedback contains only the Bernoulli signal, so the standard two-environment argument with $\varepsilon\asymp T^{-1/2}$ gives an $\Omega(T^{-1/2})$ average-regret lower bound \citep[Lemma~15.1 and Eq.~(14.12)]{Lattimore2020-xh}.}

We call an algorithm $\mL$ \emph{rate-optimal} if its worst-case external regret for stationary policies matches the minimax rate up to logarithmic factors, i.e.,
\begin{align}
\sup_{\theta_i\in\Theta, \ \mu\in\Delta(\mathcal M)} \operatorname{ER}(\mathcal L(\theta_i,T); T, \theta_i, \mu) = \tilde O(T^{-1/2}).
\end{align}

\section{Limits of standard algorithms}
\label{sec:limits_standard_algorithms}

Consistency does not by itself provide any performance guarantee for the agent when the principal is adaptive. This section shows how several standard algorithms permit the principal to extract the agent's surplus.

The following example illustrates how an adaptive principal can use a learning algorithm's actions to infer the agent's valuation and personalize subsequent prices.
\begin{eg}[A vulnerable learning rule]
\label{eg:motivating_example}
The agent has two possible valuations $\theta_1$ and $\theta_2$ with $0<\theta_1<\theta_2$, and he bids either high ($a_t=b$) or low ($a_t=a_0$) in each period. The principal allocates the good with certainty and charges $q_t$ when the agent bids high. Thus, the agent's payoffs from $b$ and $a_0$ are $\theta-q_t$ and zero, respectively. The agent uses EXP3, a canonical online learning algorithm \citep{auer2002nonstochastic}. EXP3 samples actions in proportion to their weights and shifts probability toward actions that have yielded higher payoffs. EXP3 is consistent, so if the principal's strategy is stationary with mean price $\bar q$ for high bids, the type-$\theta$ agent's expected time-average utility tends to the optimal payoff $\max\{0,\theta-\bar q\}$.

Suppose that the principal is adaptive and knows that the agent uses EXP3. The principal can fully extract the surplus of every type, regardless of her prior. The principal first sets the price $q_t=(\theta_1+\theta_2)/2$. At this price, trade yields a loss for the low type and a gain for the high type. As a result, EXP3 chooses high bids more frequently for the high type and low bids more frequently for the low type, allowing the principal to estimate the type from observed actions. The principal can infer the agent's type with high accuracy after relatively few periods. After obtaining an estimate $\hat\theta$, the principal sets $q_t=\hat\theta-\epsilon$ with small $\epsilon>0$ for the remaining periods.  The principal identifies the agent's type correctly with probability approaching one as $T$ grows. When the estimate is correct, trading at the new price yields a positive payoff, so EXP3 chooses the high bid in most remaining periods. The principal therefore extracts almost all the surplus. Theorem~\ref{thm:impossibility_noER} formalizes this argument.

\end{eg}

\subsection{No-external-regret algorithms}
\label{subsec:no_external_regret}

No external regret is a widely used learning criterion in both economics and computer science. It captures the same idea as \emph{Hannan consistency} \citep{hannan1957approximation,Hart2001-ey} or \emph{universal consistency} \citep{fudenberg1995consistency}.\footnote{
The precise definitions of Hannan consistency and no external regret vary across papers.
As we mention in the related literature section, universal consistency implies our no external regret criterion.
Hannan consistency in \citet{hannan1957approximation} and \citet{Hart2001-ey} also implies our no external regret criterion.
} For possibly non-stationary strategies of the principal, external regret compares the agent's expected time-average payoff with that of the best fixed action evaluated on the same mechanism sequence generated by actual play. Formally, given a horizon $T$, an agent's type $\theta$, and a strategy $\sigma_P^T$ for the principal, the external regret of $\mathcal L(\theta,T)$ is defined as
\begin{align}
\label{eq:ER_def}
\mathrm{ER}(\mathcal L;T,\theta,\sigma_P^T)
\coloneqq
\max_{a\in\mathcal A}
\mathop{\E}_{\mathcal L(\theta,T),\sigma_P^T}
\left[\frac{1}{T}\sum_{t=1}^T
u(a,(x_t,p_t),\theta)
-
\frac{1}{T}\sum_{t=1}^T
u(a_t,(x_t,p_t),\theta)
\right],
\end{align}
where the expectation is taken over the sequence of actions and mechanisms $(a_t,x_t,p_t)_{t=1}^T$ generated by the agent's strategy $\mathcal L(\theta,T)$ and the principal's strategy $\sigma_P^T$.

Note that when evaluating a fixed action $a$ in \eqref{eq:ER_def}, we hold fixed the law of the mechanism sequence generated by actual play. If the principal adapts, this law need not coincide with the one induced by choosing $a$ in every period under the same strategy of the principal. The benchmark does not account for how a different action history would affect future trading terms, so it need not equal the expected payoff from always choosing $a$.\footnote{\emph{Policy regret}, proposed by \citet{Arora2012-ah}, accounts for this effect by evaluating alternative fixed actions using the histories and responses they would induce under the same opponent strategy. They show that vanishing time-average policy regret is impossible against unrestricted adaptive adversaries, and that it is achievable when the opponent has bounded memory.}

The no-ER requirement demands that, uniformly over types and principal strategies, the agent's expected time-average payoff be asymptotically at least that of every fixed action evaluated along the actual mechanism sequence.
\begin{df}[No external regret (no-ER)]
\label{df:no_ER}
An online learning algorithm $\mathcal L$ has \emph{no external regret} if, for any $\theta\in\Theta$ and any sequence $(\sigma_P^T)_T$ of strategies for the principal,
\begin{align}
\label{eq:no-ER}
\limsup_{T\to\infty}
\mathrm{ER}(\mathcal L;T,\theta,\sigma_P^T) \leq 0.
\end{align}
\end{df}
The class of No-ER algorithms is rich, and there are many examples in the literature.\footnote{ Examples include exponential-weight algorithms such as EXP3 \citep{auer2002nonstochastic} and EXP3-IX \citep{neu2015explore}, regularized learning algorithms such as Tsallis-INF \citep{zimmert2021tsallis}, and bandit versions of follow-the-perturbed-leader \citep{neu2016importance}. } Some algorithms satisfy stronger regret guarantees that imply no-ER. EXP3 is a canonical example of a no-ER algorithm.
\begin{eg}[EXP3]
    \label{eg:exp3}
    In each period, $\textsf{EXP3}$ samples an action in proportion to its current weight. Its updates make actions with better past performance more likely to be selected. Because it only observes the payoff of the chosen action, it updates weights using importance-weighted estimates whose conditional expectation equals the true per-round payoff. $\textsf{EXP3}$ satisfies no-ER and is rate-optimal.
    Appendix~\ref{app:exp3} provides more details.
\end{eg}

No-ER is a popular learning criterion because it implies consistency.\footnote{This holds since, when the principal's policy is stationary, \eqref{eq:ER_def} coincides with \eqref{eq:wer_definition} and the external regret is nonnegative.} Moreover, with the opt-out action, it implies approximate individual rationality, i.e., the agent's expected payoff is asymptotically at least zero for every type and every strategy of the principal.
\begin{prop}
\label{prop:noer_asymptotic_individual_rationality}
Every no-ER algorithm is consistent. Moreover, it is asymptotically individually rational, i.e., for any $i$ and any sequence of strategies for the principal,
\begin{align}
\liminf_{T\to\infty} \bar u_{i,T} \geq 0.
\end{align}
\end{prop}

An adaptive principal can exploit no-ER algorithms. The no-ER criterion benchmarks performance against the best fixed action for the \emph{realized} sequence of mechanisms. Consequently, if the principal holds a particular mechanism fixed for sufficiently many periods, a no-ER algorithm must play an approximate best response to that mechanism.
An adaptive principal can use this implication to run a short ``probing'' phase that makes the agent's play informative about $\theta_i$, infer the agent's type from the resulting action frequencies, and then adjust future mechanisms to extract the agent's information rent.

\begin{thm}[Uniform full-surplus extraction from no-ER algorithms]
\label{thm:impossibility_noER}
Fix any no-ER algorithm and write $\bar p_{i,T}(\sigma_P^T)$ for the average payment of type $i$ under the principal's strategy $\sigma_P^T$. For any $\epsilon>0$ and any sufficiently large horizon $T$, there is a strategy $\widetilde\sigma_P^T$ for the principal such that for every type $\theta_i$, we have
\begin{align}
\label{eq:noer_uniform_extraction}
\bar p_{i,T}(\widetilde\sigma_P^T)\geq\theta_i-\epsilon.
\end{align}
\end{thm}

To prove the theorem, we extend the test in the two-type example to the general case. Let $n\coloneqq|\Theta|$. The principal runs $n-1$ test blocks in descending price order. In block $k$, she allocates the good with certainty at price $q_k=(\theta_{n-k}+\theta_{n-k+1})/2$ whenever the agent chooses the fixed purchase action. All other actions yield no trade. With appropriately chosen block lengths, types above the posted price buy in a majority of periods with high probability, while types below it do not, allowing the principal to identify the type. To establish this, we compare the actual policy with an auxiliary policy that runs the same tests through the current block and then stops trading. The no-ER guarantee bounds expected missed purchases for higher types by comparison with always buying, and expected purchases for lower types by comparison with always opting out.
Markov's inequality and a union bound make the probability of misclassifying the type as small as desired, while the total number of test periods remains sublinear in $T$.
After testing, the principal sets a price slightly below the estimated valuation. If she correctly identifies the agent's type, the expected number of missed purchases after testing is sublinear.
Making the probability of misclassifying the agent's type and the price discount sufficiently small yields the claimed revenue bound.

Asymptotic individual rationality (Proposition~\ref{prop:noer_asymptotic_individual_rationality}) and full surplus extraction (Theorem~\ref{thm:impossibility_noER}) pin down the outcomes under every asymptotically optimal response: no-ER preserves no information rent, and the principal obtains the entire surplus.
\begin{cor}
\label{cor:noer_zero_utility}
For any no-ER algorithm $\mL$ and any full-support prior $\pi$, we have
\[
\lim_{T\to\infty}V_T(\pi,\mathcal L)=\sum_i\pi_i\theta_i.
\]
Along every sequence of strategies of the principal in $\operatorname{AO}(\pi,\mathcal L)$, for every type $\theta_i\in\Theta$, we have $\bar p_{i,T}\rightarrow\theta_i$ and $\bar u_{i,T}\rightarrow0$; moreover, $\bar x_{i,T}\to1$ for every positive type.
\end{cor}

\begin{rem}
    Theorem~6 of \citet{Brown2023-pc} gives a general method for learning an approximately optimal Stackelberg strategy in finite games satisfying their regularity condition. Their underlying idea is to use the learner's no-regret behavior to simulate best-response queries. Specialized to our environment and followed by an exploitation phase, their result yields asymptotic full-surplus extraction against \emph{anytime} no-ER learning algorithms, a subclass of no-ER learning algorithms.\footnote{\citet[Definition~6]{Brown2023-pc} call a no-ER learner \emph{anytime} if, for some $c<1$, its cumulative external regret over the first $t$ periods is $O(t^c)$ for every $t\leq T$; the regret is bounded at any time, which is not satisfied by all no-ER algorithms. By contrast, Definition~\ref{df:no_ER} requires only a vanishing upper bound on time-average external regret at the terminal horizon $T$. Equivalently, cumulative regret $T\,\mathrm{ER}$ has an $o(T)$ upper bound, without assuming a corresponding regret guarantee for the first $t<T$ periods of the same interaction.}
    
    Theorem~\ref{thm:impossibility_noER} uses the structure of our model to establish full-surplus extraction for \emph{all} no-ER learning algorithms. The proof uses the principal's ability to offer no trade and the agent's opt-out action to control behavior in the test blocks using the terminal regret bound.
\end{rem}

\subsection{Other standard learning algorithms}
\label{subsec:other_standard_algorithms}

Theorem~\ref{thm:impossibility_noER} does not establish the incompatibility between consistency and rent protection since no-ER is a stronger requirement than consistency. We next consider three familiar algorithms designed for stationary bandit problems that satisfy consistency but not no-ER, and show that they can also be exploited by an adaptive principal.

Our first example of a standard consistent algorithm that does not satisfy no-ER is \textsf{UCB} (Upper Confidence Bound). Let
\[
N_a(t)\coloneqq\sum_{s=1}^t\mathbf 1\{a_s=a\}
\]
denote the number of times action $a$ has been selected by the algorithm by the end of period $t$, and let $\widehat u_{i,t}(a)$ be the average utility observed when action $a$ is selected through period $t$ for type $\theta_i$.

\begin{eg}[UCB]
\label{eg:ucb}
\textsf{UCB} first samples each action once in a fixed public order. Thereafter, type $\theta_i$ chooses an action maximizing its \emph{upper confidence bound} (UCB), defined as
\begin{align}
\label{eq:ucb_index}
\widehat u_{i,t-1}(a)
+\bar B\sqrt{\frac{2\log T}{N_a(t-1)}},
\end{align}
where $\bar B\coloneqq p_{\mathrm{max}} + \max \Theta$. Ties may be resolved arbitrarily. \textsf{UCB} is consistent and rate-optimal. Appendix~\ref{app:stationary_rates} provides more details.
\end{eg}

Like no-ER algorithms, \textsf{UCB} can reveal the agent's type through its actions in early periods, which an adaptive principal can exploit. For example, for the first $H_T=\lceil\sqrt T\rceil$ periods, let the principal charge zero for every action and allocate the good if and only if the agent chooses $b$.
For every positive type, the number of opt-outs satisfies
\begin{align}
\label{eq:ucb_type_identification}
\frac{N_{a_0}(H_T)}{\log T}
\longrightarrow\frac{2\bar B^2}{\theta_i^2}.
\end{align}
The limits differ across positive types. A zero type instead samples all actions nearly equally. Since the type set is finite, we can identify the type exactly for any sufficiently large horizon. The principal can then charge a price just below the inferred valuation when the agent chooses $b$.

\begin{prop}
\label{prop:ucb_full_extraction}
Suppose that the agent uses UCB. There exists a sequence of strategies for the principal under which $\bar p_{i,T}\to\theta_i$ and $\bar u_{i,T}\to0$ for every type, and $\bar x_{i,T}\to1$ for every positive type. For every fixed full-support prior,
\[
V_T(\pi,\textsf{UCB})\longrightarrow\sum_i\pi_i\theta_i.
\]
The same payment, utility, and positive-type allocation limits hold along every response sequence in $\operatorname{AO}(\pi,\textsf{UCB})$.
\end{prop}

\textsf{UE} (Uniform Exploration) and \textsf{SE} (Successive Elimination) are two other standard consistent algorithms that do not satisfy no-ER. These algorithms are exploited by an adaptive principal through a different mechanism.
\begin{eg}[Uniform Exploration]
\label{eg:uniform_exploration}
\textsf{Uniform Exploration (UE)} samples each action in a fixed order until each has $N_T=\lceil T^{2/3}(\log T)^{1/3}\rceil$ observations. After exploration, it chooses an action with the highest empirical average utility and uses that action in every remaining period.\footnote{
    If $N_T=0$ or $KN_T>T$, let $\textsf{Uniform Exploration (UE)}$ select $a_0$ in every period. These details do not matter for the analysis, as all the results are about the case with $T$ sufficiently large.
}
Its external regret under stationary policies is $\tilde \Theta(T^{-1/3})$, so it is consistent but not rate-optimal. Appendix~\ref{app:stationary_rates} provides more details.
\end{eg}

\begin{eg}[Successive Elimination]
\label{eg:successive_elimination}
\textsf{Successive Elimination (SE)} cycles through the active actions and maintains confidence intervals for their expected utilities. It permanently eliminates an action when its upper confidence bound falls below another active action's lower confidence bound. With confidence intervals chosen appropriately, its external regret under stationary policies is $\tilde O(T^{-1/2})$, so it is consistent and rate-optimal. Appendix~\ref{app:stationary_rates} provides more details, including the definition of confidence intervals.
\end{eg}

Both rules can stop considering alternatives after favorable early observations. The principal can exploit this behavior by worsening the terms of an action after the agent becomes committed to it.

\begin{prop}
\label{prop:ue_se_extraction}
For each of \textsf{UE} and \textsf{SE}, there exists a sequence of strategies for the principal, independent of the prior, such that
\[
\lim_{T\to\infty}\bar p_{i,T}=p_{\mathrm{max}}
\qquad\text{for every }\theta_i\in\Theta\setminus\{0\}.
\]
\end{prop}

\section{Protecting information rents}
\label{sec:protecting_information_rents}

The preceding negative results show that standard consistent learning algorithms do not ensure that the agent retains his information rent when the principal adapts. We therefore introduce our second design requirement: \emph{no dynamic rent extraction}. The benchmark is the utility that the agent would receive in the corresponding static screening problem, in which the principal maximizes revenue using only her prior, without using the information revealed during the interaction.
A learning algorithm that satisfies no dynamic rent extraction guarantees that the agent retains at least this utility when the principal adapts.

\subsection{The static screening benchmark}
\label{subsec:static_screening_benchmark}

To define a benchmark, we consider a corresponding static problem, which is different from the original game with repeated interaction.
A static direct mechanism is a collection
\[
D= \bigl(x_i,p_i \bigr)_{i=1}^n
\in \bigl([0,1]\times[0,p_{\mathrm{max}}]\bigr)^n,
\]
where $x_i$ is the allocation probability and $p_i$ is the payment for type $\theta_i$.
Given $D$, type $\theta_i$'s utility from his assigned outcome is
\[
U_i(D)\coloneqq \theta_i x_i-p_i,
\]
and the principal's expected revenue under prior $\pi$ is
\[
R_\pi(D)\coloneqq\sum_{i=1}^n\pi_i p_i.
\]
Let $\mathcal F$ denote the set of incentive compatible (IC) and individually rational (IR) direct mechanisms:
\begin{align}
\mathcal F\coloneqq\left\{D:
U_i(D)\geq \theta_i x_j-p_j\ \text{for all }i,j,
\quad U_i(D)\geq0\ \text{for all }i\right\}.
\label{eq:static_feasible_set}
\end{align}
Define the optimal static revenue, the set of optimal mechanisms, and the \emph{static screening rent} for each agent's type by
\begin{align}
\operatorname{Rev}^*(\pi) \coloneqq \max_{D\in\mathcal F}R_\pi(D), \quad
\mathcal D^*(\pi) \coloneqq \argmax_{D\in\mathcal F}R_\pi(D), \quad
u^*(\theta_i,\pi) \coloneqq \min_{D\in\mathcal D^*(\pi)}U_i(D).
\label{eq:static_rent_benchmark}
\end{align}
The set $\mathcal F$ is compact and nonempty and $R_\pi$ is continuous, so these objects are well defined. When several static mechanisms maximize revenue, each type's benchmark is the lowest utility he would receive under any of those mechanisms. Note that $\mathcal D^*(\pi)$ is generically a singleton; in such cases, $u^*(\theta_i,\pi)=U_i(D^*(\pi))$ is pinned down without taking a minimum.

The static benchmark has a simple posted-price representation.\footnote{
    Posted-price optimality is classical \citep{Myerson1981,RileyZeckhauser1983}. \citet{Skreta2006} establishes it for arbitrary bounded measurable type spaces, including finite ones.
}
The direct mechanism $D^k$ associated with price $\theta_k$ assigns
\[
x_i^k \coloneqq \mathbf 1\{i\geq k\},
\quad
p_i^k \coloneqq \theta_k\mathbf 1\{i\geq k\} \qquad\text{for all }i.
\]
Let $p^*(\pi)$ denote the highest price $\theta_k$ that maximizes $R_\pi(D^k)$.

\begin{prop}
\label{prop:static_screening_characterization}
For every full-support prior $\pi$,
\begin{align}
\operatorname{Rev}^*(\pi)
=\max_{1\leq k\leq n}\theta_k\sum_{i=k}^n\pi_i.
\end{align}
The set $\mathcal D^*(\pi)$ is the convex hull of the mechanisms $D^k$ associated with revenue-maximizing posted prices.
Consequently, the static screening rent can be written as
\begin{align}
u^*(\theta_i,\pi)
=\bigl(\theta_i-p^*(\pi)\bigr)_+.
\label{eq:explicit_static_rent}
\end{align}
\end{prop}

The highest revenue-maximizing posted price minimizes every type's utility simultaneously. Thus, a single optimal static mechanism attains the entire static screening rent vector, even when multiple optimal mechanisms exist.
The characterization also shows that the optimal direct mechanism is unique if and only if the revenue-maximizing posted price is unique; multiplicity arises precisely when distinct posted prices tie for maximal revenue.

\subsection{No dynamic rent extraction}
\label{subsec:no_dynamic_rent_extraction}

The principal uses her prior to maximize her expected revenue in the static problem. In repeated trade, she can also use information revealed in earlier periods to set future terms. Our protection requirement asks that this additional source of information not reduce any type's long-run utility below his static screening rent.

\begin{df}[No dynamic rent extraction]
\label{df:no_dre}
An online learning algorithm $\mathcal L$ satisfies \emph{no dynamic rent extraction} if, for every full-support prior $\pi\in\Delta(\Theta)$, every asymptotically optimal response sequence $(\sigma_P^T)_T\in\operatorname{AO}(\pi,\mathcal L)$, and every type $\theta_i\in\Theta$,
\begin{align}
\liminf_{T\to\infty}\bar u_{i,T}
\geq u^*(\theta_i,\pi).
\label{eq:no_dre}
\end{align}
\end{df}

No dynamic rent extraction compares the agent's rent under adaptive revenue maximization with his static screening rent.
No dynamic rent extraction can guarantee strictly positive rents for some types for some priors. For example, if the principal's prior puts enough weight on the smallest positive type $\theta_+$ to make $\theta_+$ the unique optimal posted price, then the highest type's benchmark is $u^*(\theta_n,\pi)=\theta_n-\theta_+>0$. The negative results in Section~\ref{sec:limits_standard_algorithms} imply the following corollary:

\begin{cor}
\label{cor:standard_algorithms_fail_no_dre}
Every no-external-regret algorithm fails to satisfy no dynamic rent extraction, as do \textsf{UCB}, \textsf{UE}, and \textsf{SE}.
\end{cor}

\section{Robust learning algorithms}
\label{sec:robust_algorithms}

We now construct learning algorithms that satisfy consistency and no dynamic rent extraction simultaneously. Both proposed algorithms implement the same design principle: the agent first explores in a type-independent manner to form a baseline without revealing his private information to the principal. This baseline is called the \emph{empirical menu}: for each action, it records the average allocation and payment observed during exploration.
The algorithm then exploits, choosing an action that performs well under the empirical menu as long as realized outcomes remain consistent with the menu.
If outcomes drift away from the menu in a statistically detectable way, the algorithm permanently opts out.
This test-and-opt-out structure forces the principal to propose one empirical menu during exploration common to all types and to keep using it for the rest of the interaction.

We present two versions of this test-and-opt-out construction.  The first, \emph{Guarded Uniform Exploration} (\textsf{G-UE}), is deliberately simple.  The second, \emph{Guarded Successive Elimination} (\textsf{G-SE}), retains the same protection while achieving rate-optimal learning under stationarity.

We first introduce notation common to both algorithms.  Let $K\coloneqq|\mathcal A|$, $\overline\theta\coloneqq\theta_n$, and
\[
B_0\coloneqq
\max\{1,p_{\mathrm{max}},\overline\theta+p_{\mathrm{max}}\}.
\]
For $T\geq2$ and $s\geq1$, let
\begin{align}
\rho_s\coloneqq C_0\sqrt{\frac{\log T}{s}},
\qquad C_0\coloneqq \sqrt{2}B_0.
\label{eq:robust_confidence_radius}
\end{align}

For expositional simplicity, when defining the algorithms below, we assume that $T$ is sufficiently large so that the full exploration schedule fits. This assumption is not restrictive, as all the results are about the case with $T$ sufficiently large.

\subsection{A simple robust algorithm: \textsf{G-UE}}
\label{subsec:ue_trigger}

Set
$
N_T\coloneqq
\left\lceil T^{2/3}(\log T)^{1/3}\right\rceil
.
$
Algorithm \textsf{G-UE} has the following three phases.

\paragraph{Phase 1: type-independent exploration.}
The agent cycles through all actions until each action has been selected exactly $N_T$ times.
The empirical menu is $(\overline x_a,\overline p_a)_{a\in\mathcal A}$, where $\overline x_a$ and $\overline p_a$ are the average allocation and payment observed when action $a$ was selected during exploration.
Type $\theta_i$ then selects an action
\[
\widehat a_i\in
\argmax_{a\in\mathcal A}
\{\theta_i\overline x_a-\overline p_a\}.
\]

\paragraph{Phase 2: exploitation with protection.}
The agent repeatedly selects $\widehat a_i$ during this phase. After the agent selects that action for the $s$-th time in Phase 2, let $\overline x^{(2)}_{\widehat a_i,s}$ and $\overline p^{(2)}_{\widehat a_i,s}$ be the Phase-2 empirical means. The algorithm tests consistency separately for allocation and payment. Writing $y=x$ for allocation and $y=p$ for payment, it compares the Phase-1 and Phase-2 intervals
\begin{align}
\left[\overline y_{\widehat a_i}-\rho_{N_T},
      \overline y_{\widehat a_i}+\rho_{N_T}\right]
\quad\text{and}\quad
\left[\overline y^{(2)}_{\widehat a_i,s}-\rho_s,
      \overline y^{(2)}_{\widehat a_i,s}+\rho_s\right].
\label{eq:ue_trigger_test}
\end{align}
A consistency test passes if its two intervals overlap and fails otherwise. If either the allocation or payment consistency test fails, the algorithm enters Phase 3.

\paragraph{Phase 3: opt out.}
The agent selects $a_0$ in every subsequent period.

Under stationarity, with high probability the intervals in \eqref{eq:ue_trigger_test} contain the same time-invariant means, so all consistency tests pass.  On that event, \textsf{G-UE} behaves like ordinary \textsf{UE} and inherits its consistency.  Against an adaptive principal, in contrast, passing the consistency tests forces the Phase-2 allocation and payment totals to remain close to those generated by the Phase-1 empirical menu.  Because the selected action is an exact best response to that menu, the empirical menu induces an incentive compatible and individually rational outcome across types.

\subsection{A rate-optimal robust algorithm: \textsf{G-SE}}
\label{subsec:ue_se_trigger}

The long exploration phase of \textsf{G-UE} yields an external regret bound of $\widetilde \Theta(T^{-1/3})$, which is slower than the minimax rate of $\widetilde O(T^{-1/2})$.
The next algorithm, \textsf{G-SE}, improves this rate by separating the baseline formation for rent protection from efficient learning.

Set $m_T\coloneqq\lceil\sqrt T\rceil$.  In Phase~1, the agent cycles through the actions until each has been selected $m_T$ times and records the empirical menu $(\overline x_a,\overline p_a)_{a\in\mathcal A}$ as in \textsf{G-UE}. Phase~2 begins with every action active.  It proceeds in elimination rounds of \textsf{SE}: each active action is selected once, in a fixed order, and eliminations are made simultaneously only at the end of a completed round.

Suppose action $a$ has been selected $s_a$ times in Phase~2. Let $\widehat x_{a,m_T+s_a}$ and $\widehat p_{a,m_T+s_a}$ denote the empirical means of allocation and payment, computed from all $m_T+s_a$ observations of $a$ across both phases. Whenever $a$ is selected, the algorithm performs the allocation and payment consistency tests by comparing, for $y\in\{x,p\}$,
\begin{align}
\left[\overline y_a-\rho_{m_T},
      \overline y_a+\rho_{m_T}\right]
\quad\text{and}\quad
\left[\widehat y_{a,m_T+s_a}-\rho_{m_T+s_a},
      \widehat y_{a,m_T+s_a}+\rho_{m_T+s_a}\right].
\label{eq:ue_se_trigger_test}
\end{align}
If either consistency test fails, the algorithm selects $a_0$ in all subsequent periods.

If all consistency tests pass throughout a completed elimination round, define for each action $a$ active in that round
\begin{align*}
\widehat u_i(a)
&\coloneqq
\theta_i\widehat x_{a,m_T+s_a}
-\widehat p_{a,m_T+s_a},\\
\operatorname{LCB}_i(a)
&\coloneqq\widehat u_i(a)-\rho_{m_T+s_a},
\qquad
\operatorname{UCB}_i(a)
\coloneqq\widehat u_i(a)+\rho_{m_T+s_a}.
\end{align*}
The algorithm simultaneously eliminates every active action $a$ satisfying
\begin{align}
\operatorname{UCB}_i(a)
<\max_{a'\in\mathcal A\colon \mathrm{active}}\operatorname{LCB}_i(a').
\label{eq:ue_se_elimination_rule}
\end{align}
It then begins another round with the surviving actions.  The round-based formulation ensures that the sample counts of active actions differ by at most one.

Under stationarity, with high probability the baseline and updated intervals contain the corresponding true means, so all consistency tests pass. On that event, Phase~2 runs successive elimination, starting with $m_T$ observations per action, and inherits the rate-optimal regret bound of \textsf{SE}. Under adaptation, passing the consistency tests bounds deviations from the Phase-1 empirical menu. Together with the elimination rule, this ensures approximate best-response behavior relative to the menu, up to errors that vanish in time-average payoffs.

\subsection{Learning and rent-protection guarantees}
\label{subsec:robust_guarantees}

The next theorem collects the properties of the two proposed algorithms.

\begin{thm}[Robust learning]
\label{thm:robust_algorithms}
Fix the finite action and type sets.
\begin{enumerate}
\item The algorithms \textsf{G-UE} and \textsf{G-SE} both satisfy consistency and no dynamic rent extraction. Moreover, \textsf{G-SE} is rate-optimal.

\item For either algorithm $\mathcal L$ and every full-support prior $\pi$,
\begin{align}
V_T(\pi,\mathcal L)\longrightarrow\operatorname{Rev}^*(\pi).
\label{eq:revenue_convergence_ue_se}
\end{align}
\item \textsf{G-UE} and \textsf{G-SE} are asymptotically individually rational, i.e., against an arbitrary sequence of strategies for the principal, both algorithms guarantee $\liminf_{T\to\infty}\overline u_{i,T}\geq0$ for every type.
\end{enumerate}
\end{thm}

We sketch the proof of Theorem~\ref{thm:robust_algorithms}.
For consistency, concentration bounds imply that the confidence intervals contain the corresponding stationary means with high probability. On this event, all consistency tests pass for both algorithms: \textsf{G-UE} follows \textsf{UE}, and \textsf{G-SE} follows \textsf{SE}. Thus, the consistency of \textsf{UE} and \textsf{SE} is inherited by the algorithms.

Suppose the principal is adaptive. Type-independent exploration in Phase 1 reveals no information about the agent's type, and thus the principal proposes a single empirical menu common to all types.
For \textsf{G-UE}, each type chooses the optimal action from the empirical menu. Assigning each type the corresponding outcome defines an IC static direct mechanism. The opt-out option ensures IR.
For \textsf{G-SE}, each type chooses a mixture of approximately optimal actions under the empirical menu. Assigning each type the corresponding outcome defines an approximately IC and IR static direct mechanism for each type.

The consistency tests ensure that the principal's deviations from the empirical menu are quickly detected, and such deviations reduce the agent's cumulative utility before opt-out by at most $o(T)$. After opt-out, utility is zero. Phase-1 exploration also contributes only $o(T)$. Since his choices are (approximately) IR, these imply that his average utility is at least $-o(1)$ against any principal strategy, proving asymptotic individual rationality.
As the consistency tests prevent the principal from earning more revenue than she would earn under the proposed empirical menu, it is approximately optimal for the principal not to let the agent opt out after Phase 1. (Approximate) IC and IR therefore give the asymptotic revenue upper bound $\operatorname{Rev}^*(\pi)$. Conversely, consistency of the algorithms lets the principal attain this bound asymptotically by posting a price slightly below an optimal static price. Together, these prove \eqref{eq:revenue_convergence_ue_se}.
An asymptotically optimal principal, facing the consistency tests, must choose a mechanism that is close to an optimal static mechanism. Proposition~\ref{prop:static_screening_characterization} shows that every optimal static mechanism is a mixture of optimal posted-price mechanisms and gives each type utility at least $u^*(\theta_i,\pi)$. Since letting the agent opt out early only reduces the principal's revenue, asymptotic optimality preserves these rents in time average. This proves no dynamic rent extraction.

\section{Optimality of rent protection}
\label{sec:optimality}

We have constructed algorithms that satisfy consistency and no dynamic rent extraction. We now show that the proposed algorithms are Pareto optimal among consistent learning algorithms with respect to their rent guarantees across type--prior pairs: improving the guarantee at one type--prior pair requires reducing it at another type--prior pair.

For a full-support prior $\pi$, define type $\theta_i$'s worst-response rent guarantee under $\mathcal L$ by
\begin{align}
g_i(\pi,\mathcal L)
\coloneqq
\inf_{(\sigma_P^T)_T\in\operatorname{AO}(\pi,\mathcal L)}
\liminf_{T\to\infty}\bar u_{i,T},
\label{eq:worst_response_rent_guarantee}
\end{align}
where each utility sequence is evaluated under the corresponding sequence of strategies for the principal.
With this notation, no dynamic rent extraction is equivalent to $g_i(\pi,\mathcal L)\geq u^*(\theta_i,\pi)$ for every type $i$ and every full-support prior $\pi$.

Our goal is to show that, for every consistent algorithm satisfying no dynamic rent extraction, each type's worst-response rent guarantee equals his static screening rent. This implies the Pareto optimality of the proposed algorithms with respect to their rent guarantees. To this end, we first show that no dynamic rent extraction caps the principal's asymptotic revenue at the static optimal revenue.

\begin{prop}
\label{prop:no_dre_revenue_ceiling}
If $\mathcal L$ satisfies no dynamic rent extraction, then, for every full-support prior $\pi$,
\begin{align}
\limsup_{T\to\infty}V_T(\pi,\mathcal L)
\leq \operatorname{Rev}^*(\pi).
\label{eq:no_dre_revenue_ceiling}
\end{align}
\end{prop}

We sketch the proof of Proposition~\ref{prop:no_dre_revenue_ceiling}.
Prior robustness required in the definition of no dynamic rent extraction is essential to this bound.
Fix any prior $q$ under which posting the lowest positive valuation $\theta_+$ is the unique optimal static price. Under such $q$, each positive type's static screening rent is his valuation minus $\theta_+$. No dynamic rent extraction requires that, under any revenue-maximizing response, he pay at most $\theta_+$ on average. Her revenue in the original game therefore cannot exceed the static optimum at prior $q$.
Now suppose toward contradiction that, at some prior $\pi$, the principal could earn revenue $R(\pi)$ strictly above the optimal static revenue $\operatorname{Rev}^*(\pi)$.
Starting from $\pi$, increase the weight on $\theta_+$ until posting $\theta_+$ just becomes the unique optimal static price.
Denote this new prior by $q_\alpha  \coloneqq \alpha \pi + (1-\alpha) \delta_{\theta_+}$.\footnote{$\delta_{\theta_+}$ denotes the point mass at $\theta_+$.}
At $q_\alpha$, in the static problem, the principal is nearly indifferent between posting $\theta_+$ and posting the optimal price under $\pi$, and the type-$\theta_+$ agent will not buy at that price. The static optimum at $q_\alpha$ is therefore approximately $\alpha \operatorname{Rev}^*(\pi)$.
Under $q_\alpha$, the principal's revenue-maximizing response for $\pi$ earns at least $\alpha R(\pi)$. This exceeds the optimal static revenue under $q_\alpha$, a contradiction.

Consistency supplies the matching revenue lower bound. The next lemma shows that the principal can replicate the outcomes of any static posted-price mechanism in the original game.

\begin{lem}
\label{lem:stationary_cutoff_outcomes}
Suppose $\mathcal L$ is consistent.
For every positive $p\in\Theta$, there is a sequence of deterministic stationary principal strategies such that, for every $i$,
\begin{align}
(\bar x_{i,T},\bar p_{i,T},\bar u_{i,T})
\longrightarrow
\bigl(\mathbf 1\{\theta_i\geq p\},
p\mathbf 1\{\theta_i\geq p\},
(\theta_i-p)_+\bigr).
\label{eq:stationary_cutoff_outcomes}
\end{align}
\end{lem}

By Proposition~\ref{prop:static_screening_characterization}, the highest revenue-maximizing posted price $p^*(\pi)$ yields revenue $\operatorname{Rev}^*(\pi)$ and typewise utilities $u^*(\theta_i,\pi)$. Applying Lemma~\ref{lem:stationary_cutoff_outcomes} at this price gives a sequence of stationary strategies under which revenue and utilities for each type converge to these values for any consistent algorithm $\mathcal L$.
If $\mathcal L$ also satisfies no dynamic rent extraction, Proposition~\ref{prop:no_dre_revenue_ceiling} implies that the principal's value $V_T(\pi,\mathcal L)$ converges to the same revenue. The stationary response sequence is therefore asymptotically optimal and achieves the infimum in the definition of $g_i(\pi,\mathcal L)$.
These observations yield the following theorem.

\begin{thm}[Exact rent guarantees]
\label{thm:sharpness_rent_protection}
Suppose $\mathcal L$ is consistent and satisfies no dynamic rent extraction. Then, for every full-support prior $\pi$, the adaptive principal's value converges to the optimal static revenue:
\[
V_T(\pi,\mathcal L)\longrightarrow\operatorname{Rev}^*(\pi).
\]
Moreover, every type's worst-response rent guarantee equals his static screening rent:
\[
g_i(\pi,\mathcal L)=u^*(\theta_i,\pi).
\]
\end{thm}

Any consistent algorithm that weakly improves the guarantees of $\textsf{G-UE}$ or $\textsf{G-SE}$ for every type and full-support prior would itself satisfy no dynamic rent extraction. Theorem~\ref{thm:sharpness_rent_protection} would then force its guarantees to equal the same static screening rents, excluding any strict improvement.

\begin{cor}[Pareto optimality across type--prior pairs]
\label{cor:pareto_optimal_rent_protection}
Each of $\textsf{G-UE}$ and $\textsf{G-SE}$ is Pareto optimal among consistent algorithms with respect to rent guarantees across type--prior pairs. That is, for either algorithm $\mathcal L^*\in\{\textsf{G-UE},\textsf{G-SE}\}$, there is no consistent algorithm $\widetilde{\mathcal L}$ such that for every $i$ and every full-support prior $\pi$, we have
\[
g_i(\pi,\widetilde{\mathcal L})
\geq g_i(\pi,\mathcal L^*)
\]
with a strict inequality for at least one type--prior pair.
\end{cor}

Consider an algorithm developer who wants to design a general-purpose learning algorithm that protects users' information rents. The developer must commit to the algorithm without knowing who will use it for what. In such a setting, rather than focusing on its performance for a specific user or platform, the developer naturally cares about the algorithm's performance across different user types and different platforms. If consistency is taken as a minimal performance requirement, the corollary means that the proposed algorithms offer maximal protection in the Pareto sense: no consistent algorithm can raise the guarantee for some type-prior pair without lowering it for another type-prior pair.
\section{Extension: learning with unknown values}
\label{sec:extensions}

The baseline model assumes that the agent knows his valuation (e.g., willingness to pay for a good or service). In many applications, however, the agent must learn it from experience.
Delegating decisions to a learning algorithm is especially natural in this setting: the algorithm uses observed realized values, in addition to allocations and payments, to learn which actions yield higher payoffs.
Examples include an advertiser who learns the value of impressions from subsequent conversions and a buyer who learns a product's value through consumption.
We model this setting by drawing values from a fixed but unknown distribution independently across periods.

\paragraph{Model with unknown i.i.d.~values.}
Let $\Theta_\mu=\{\mu_1<\cdots<\mu_n\}\subset[0,p_{\mathrm{max}})$ be a finite set of possible means of the value distribution containing at least two positive elements. Conditional on a fixed mean $\mu_i$, the per-period values $\theta_t$ are i.i.d.\ with distribution $F_i$ on $[0,B]$, where $B<\infty$ is a known bound and $\int \theta\,dF_i(\theta)=\mu_i$. Neither player initially observes the mean. The principal has a full-support prior $\pi$ over $\Theta_\mu$. The agent need not know this prior or the family $(F_i)_i$. When comparing guarantees across priors, we hold $(F_i)_i$ fixed.

The action set, mechanisms, and simultaneous moves in each stage game are as in the baseline model. In period $t$, after both players move, $\theta_t$ is drawn independently of past values and the players' randomization, conditional on $\mu_i$. The agent observes the allocation and payment for his chosen action. If he receives the good, he also privately observes $\theta_t$.
The principal observes the agent's actions and her own mechanisms, but cannot observe the value draws. In particular, the realized reward
\[
r_t=\theta_t x_t(a_t)-p_t(a_t)
\]
is always observed by the agent.
The algorithm takes the horizon as input. It does not require knowledge of the true mean $\mu$ of the agent's value distribution.

Writing $\E_i$ for expectation conditional on $\mu_i$, we have
\begin{align}
\E_i[r_t\mid\text{past},a_t,x_t,p_t]
=\mu_i x_t(a_t)-p_t(a_t),
\qquad
\bar u_{i,T}&=\mu_i\bar x_{i,T}-\bar p_{i,T}.
\label{eq:iid_payoff_identity}
\end{align}
The interim payoff identity therefore has the same form as in the baseline model, with $\mu_i$ in place of $\theta_i$. Let $\mathcal L(T)$ be the strategy generated by the algorithm. Its external regret given stationary policy $\nu$ is $\operatorname{ER}(\mathcal L(T);T,\mu_i,\nu)$. Consistency requires this regret to vanish uniformly over means and stationary policies.
By \eqref{eq:iid_payoff_identity}, the external regret just defined is also the expected time-average fixed-action regret computed from the realized rewards $\theta_t x_t(a_t)-p_t(a_t)$.

We can define $V_T$, asymptotically optimal responses, no dynamic rent extraction, and the worst-response rent guarantees $g_i$ analogously. The static benchmark is the screening problem indexed by the means: replace $\theta_i$ by $\mu_i$ in \eqref{eq:static_feasible_set}--\eqref{eq:static_rent_benchmark}. In particular, if $p^*(\pi)$ is the highest revenue-maximizing posted price,
\[
\operatorname{Rev}^*(\pi)=\max_k\mu_k\sum_{i\geq k}\pi_i,
\qquad
u^*(\mu_i,\pi)=(\mu_i-p^*(\pi))_+.
\]

\paragraph{Adapting the proposed algorithms.}
Both \textsf{G-UE} and \textsf{G-SE} can learn directly from observed rewards. Keep the deterministic exploration schedule and the allocation and payment consistency tests. Replace every utility estimate used to select or eliminate an action with the empirical mean of that action's observed rewards.
For this extension, let
\[
\rho_s \coloneqq C_{\mathrm{iid}}\sqrt{\frac{\log T}{s}},\qquad
C_{\mathrm{iid}} \coloneqq \sqrt{2}(B+p_{\mathrm{max}}+1).
\]
Use the exploration lengths and the same large-horizon convention as in Section~\ref{sec:robust_algorithms}.

Denote these versions by $\textsf{G-UE}^{\mathrm{iid}}$ and $\textsf{G-SE}^{\mathrm{iid}}$.

\begin{thm}[Learning and rent protection with unknown i.i.d.\ values]
\label{thm:unknown_iid_extension}
Fix the finite action set $\mathcal A$ and mean set $\Theta_\mu$ in the model with unknown i.i.d.~values.
\begin{enumerate}
\item For every no-ER algorithm $\mathcal L$, every $\epsilon>0$, and every sufficiently large horizon $T$, there exists a principal strategy $\widetilde\sigma_P^T$ such that
\[
\bar p_{i,T}(\widetilde\sigma_P^T)\geq\mu_i-\epsilon
\qquad\text{for every }\mu_i\in\Theta_\mu.
\]

\item The algorithms $\textsf{G-UE}^{\mathrm{iid}}$ and $\textsf{G-SE}^{\mathrm{iid}}$ both satisfy consistency and no dynamic rent extraction.
Both algorithms are asymptotically individually rational.

\item For any consistent algorithm $\mathcal L$ satisfying no dynamic rent extraction and any full-support prior $\pi$ over $\Theta_\mu$, we have
\[
V_T(\pi,\mathcal L)\longrightarrow\operatorname{Rev}^*(\pi),
\qquad
g_i(\pi,\mathcal L)=u^*(\mu_i,\pi)
\quad\text{for every }\mu_i\in\Theta_\mu.
\]
Consequently, $\textsf{G-UE}^{\mathrm{iid}}$ and $\textsf{G-SE}^{\mathrm{iid}}$ are Pareto optimal among consistent algorithms with respect to their rent guarantees across type--prior pairs.
\end{enumerate}
\end{thm}

The main new issue is whether an adaptive principal can bias the Phase~1 reward estimates the agent uses to compare actions.
The principal's allocation decisions determine which value draws the agent observes. However, she cannot make those observations systematically overstate or understate the payoffs of the menu she offers in Phase~1.
The reason is that Phase-1 actions are fixed in advance, the principal never observes the realized values, and each value draw is independent of that period's allocation and payment. Thus, the principal cannot allocate selectively when the value happens to be high.
For each action $a$, the average reward in Phase~1 therefore concentrates around the utility $\mu_i\bar x_a-\bar p_a$ offered by the empirical menu. Withholding the good can limit what the agent learns about $\mu_i$, but it does not systematically bias his Phase~1 estimate of the menu payoff.

The agent therefore need not estimate $\mu_i$ separately; the algorithm can compare actions directly through their average rewards. It consequently selects an approximate best response to the Phase-1 menu, and the consistency tests prevent the principal from substantially changing the allocation and payment attached to that choice in Phase~2.
The remaining claims follow from the corresponding known-type arguments with means in place of types.

\section{Conclusion}
\label{sec:conclusion}

This paper shows that rate-optimal learning against a stationary policy is compatible with protecting an agent's information rents against an adaptive, revenue-maximizing principal. While standard learning rules---including all no-external-regret algorithms and canonical bandit algorithms---leave a privately informed agent vulnerable to full surplus extraction or worse, anchoring play to a type-independent empirical menu and enforcing it via a stationarity test guarantees every type at least his static screening rent at no asymptotic cost to the regret rate under stationarity. Moreover, among all consistent algorithms, the proposed algorithms are Pareto optimal with respect to their rent guarantees across type--prior pairs, and the results extend to settings in which the agent must simultaneously learn his own mean valuation.

More broadly, when learning agents with private information operate in environments that may adapt to their behavior, standard regret-based criteria are incomplete. Algorithm design in such settings calls for criteria that explicitly account for strategic policy adaptation and are anchored to economically meaningful benchmarks. Developing such criteria for markets with multiple learning agents, different payoff structures, or different information structures remains an important direction for future research.

\newpage
\bibliography{references}
\bibliographystyle{aer}

\appendix
\renewcommand{\thethm}{A.\arabic{thm}}
\renewcommand{\thelem}{A.\arabic{lem}}
\renewcommand{\theprop}{A.\arabic{prop}}
\renewcommand{\thecor}{A.\arabic{cor}}
\renewcommand{\thefact}{A.\arabic{fact}}
\renewcommand{\theobs}{A.\arabic{obs}}
\renewcommand{\theclaim}{A.\arabic{claim}}
\renewcommand{\theass}{A.\arabic{ass}}
\renewcommand{\therem}{A.\arabic{rem}}
\renewcommand{\thedf}{A.\arabic{df}}
\renewcommand{\theeg}{A.\arabic{eg}}
\renewcommand{\theconj}{A.\arabic{conj}}
\renewcommand{\thequestion}{A.\arabic{question}}

\setcounter{thm}{0}
\setcounter{lem}{0}
\setcounter{prop}{0}
\setcounter{cor}{0}
\setcounter{fact}{0}
\setcounter{obs}{0}
\setcounter{claim}{0}
\setcounter{ass}{0}
\setcounter{rem}{0}
\setcounter{df}{0}
\setcounter{eg}{0}
\setcounter{conj}{0}
\setcounter{question}{0}

\clearpage
\section*{Appendix}

\renewcommand{\theHthm}{appendix.\arabic{thm}}
\renewcommand{\theHlem}{appendix.\arabic{lem}}
\renewcommand{\theHprop}{appendix.\arabic{prop}}
\renewcommand{\theHcor}{appendix.\arabic{cor}}

The notation $R_T^{\mathrm{cum}}$ below denotes a cumulative regret upper bound, not the time-average regret. Constants may depend on the fixed action and type sets and payoff bounds.

\section{Learning rules and stationary regret}

\subsection{A loss formulation of EXP3}
\label{app:exp3}

Here is a precise version of Example~\ref{eg:exp3}. Put $K=|\mathcal A|$, initialize $w_1(a)=1$, and set $\eta=\sqrt{2\log K/(KT)}$. For type $\theta_i$, let $B_i=\theta_i+p_{\mathrm{max}}$. In period $t$, sample $a_t$ with probabilities $q_t(a)=w_t(a)/\sum_b w_t(b)$, observe utility $u_t$, and update
\[
\ell_t=\frac{\theta_i-u_t}{B_i},\qquad
\widehat\ell_t(a)=\frac{\ell_t\mathbf1\{a_t=a\}}{q_t(a)},\qquad
w_{t+1}(a)=w_t(a)\exp(-\eta\widehat\ell_t(a)).
\]
The losses $\ell_t(a)=(\theta_i-u(a,(x_t,p_t),\theta_i))/B_i$ lie in $[0,1]$, and only the chosen loss is observed. Since the current mechanism and action draws are conditionally independent given the past, the principal is non-anticipating.
The pseudo-regret bound of \citet[Theorem~3.1]{bubeck2012regret}, rescaled by $B_i/T$, gives
\[
\mathrm{ER}(\textsf{EXP3};T,\theta_i,\sigma_P^T)
\leq B_i\sqrt{\frac{2K\log K}{T}}
\leq 2p_{\mathrm{max}}\sqrt{\frac{2K\log K}{T}}.
\]
This bound also holds for adaptive non-anticipating principals \citep[Section~11.5, Note~4]{Lattimore2020-xh}; its uniformity over types and policies establishes no-ER and hence consistency.

\subsection{Successive elimination and the stationary bounds}
\label{app:stationary_rates}

To specify Example~\ref{eg:successive_elimination}, initially make all actions active. In each round, select every active action once in a fixed order. At the end of each completed round, use the confidence interval $[\widehat u(a)-r_s,\widehat u(a)+r_s]$ for each active action $a$ with $s$ observations, where $r_s=\bar B\sqrt{2\log T/s}$ and $\bar B=\theta_n+p_{\mathrm{max}}$. Simultaneously eliminate the actions whose upper endpoint is strictly below the largest lower endpoint among active actions. Eliminations are permanent. The UCB rule is already specified by \eqref{eq:ucb_index}.

Under stationarity, each action generates i.i.d.\ utility observations, and $(u+p_{\mathrm{max}})/\bar B\in[0,1]$. The UE (called Explore-First in \citet{Slivkins2019-ww}), SE, and UCB bounds of \citet[Theorems~1.5, 1.10, and~1.15]{Slivkins2019-ww} therefore apply after rescaling. With fixed $K$, the exploration length in Example~\ref{eg:uniform_exploration} gives average regret $O(T^{-1/3}(\log T)^{1/3})$, while SE and UCB give $O(\sqrt{K\log T/T}+K/T)$, uniformly over stationary policies and types.\footnote{The theorem numbers refer to the April~2024 online version of \citet{Slivkins2019-ww}.}

For the guarded algorithms, it remains to control failures of the consistency tests under stationarity and the cost of initial exploration. Fix a type $\theta_i$ and a stationary policy $\mu$.
Let $g(a) \coloneqq \E_\mu[u(a,(x,p),\theta_i)]$ and $\Delta(a) \coloneqq \max_b g(b)-g(a)$.

For each action, construct an infinite i.i.d.\ sequence of allocation--payment pairs with that action's marginal distribution under $\mu$. Take the sequences to be independent across actions. Whenever an action is chosen, reveal the next pair in its sequence. This construction preserves the observed-history law because only the chosen outcome is observed and each mechanism draw is fresh \citep[Section~1.3.1]{Slivkins2019-ww}.

For $y\in\{x,p,u\}$, let $Y^y_{a,j}$ denote the allocation ($y=x$), payment ($y=p$), or utility ($y=u$) in action $a$'s $j$-th observation, with utility computed at type $\theta_i$, and write $m_y(a)=\E[Y^y_{a,j}]$. Consider the two averages
\[
\overline Y^y_{a,s}=\frac1s\sum_{j=1}^sY^y_{a,j},
\qquad
\overline Y^{y,\mathrm{post}}_{a,s}
=\frac1s\sum_{j=N_T+1}^{N_T+s}Y^y_{a,j}.
\]
The first covers the exploration and all-sample estimates; the second covers \textsf{G-UE}'s Phase-2 estimates, which exclude its exploration observations.

Each average uses $s$ independent observations taking values in an interval of length at most $B_0$. For either average, denoted here by $\overline Y$, Hoeffding's inequality \citep[Corollary~5.5]{Lattimore2020-xh} gives
\[
\Pr\bigl(|\overline Y-m_y(a)|>\rho_s\bigr)
\leq2\exp\left(-\frac{2s\rho_s^2}{B_0^2}\right)
=2T^{-4},
\]
where the equality uses $\rho_s=C_0\sqrt{\log T/s}$ and $C_0=\sqrt2B_0$ from \eqref{eq:robust_confidence_radius}.

Let $\mathcal E_T$ be the event that both averages are within $\rho_s$ of $m_y(a)$ for every action $a$, every $y\in\{x,p,u\}$, and every sample count $1\leq s\leq T$. There are $2\cdot K\cdot3\cdot T=6KT$ averages, so a union bound gives
\begin{equation}
\Pr(\mathcal E_T)\geq1-(6KT) \cdot (2T^{-4})=1-12KT^{-3}.
\label{eq:app_stationary_concentration}
\end{equation}

On $\mathcal E_T$, both intervals in every consistency test contain the same stationary mean, so all consistency tests pass for both algorithms. Each utility gap $\Delta(a)$ is at most $\bar B$, so the complement of $\mathcal E_T$ contributes at most $\bar B\Pr(\mathcal E_T^c)=O(KT^{-3})$ to expected average regret.

On $\mathcal E_T$, \textsf{G-UE} follows Explore-First. Applying the standard Explore-First argument with exploration length $N_T$ (see \citet[Theorem~1.5]{Slivkins2019-ww}) and adding the contribution from $\mathcal E_T^c$ gives the following bound.
\begin{align}
\operatorname{ER}(\textsf{G-UE}(\theta_i,T);T,\theta_i,\mu)
&=O\bigl(T^{-1/3}(\log T)^{1/3}\bigr).
\label{eq:app_gue_rate}
\end{align}

On the same event, Phase~2 of \textsf{G-SE} is Successive Elimination initialized with $m_T$ observations per action. After $s\geq1$ observations of an action in Phase~2, the confidence radius for that action satisfies $\rho_{m_T+s}\leq C_0\sqrt{\log T/s}$. Thus the argument in \citet[proof of Theorem~1.10]{Slivkins2019-ww} gives cumulative Phase-2 regret $O(\sqrt{KT\log T}+K)$.
Adding the initial-exploration cost $\bar B K m_T$, dividing by $T$, and adding the contribution from $\mathcal E_T^c$ gives
\begin{align}
\operatorname{ER}(\textsf{G-SE}(\theta_i,T);T,\theta_i,\mu)
&=O\left(\frac{Km_T}{T}+\sqrt{\frac{K\log T}{T}}\right).
\label{eq:app_gse_rate}
\end{align}
Both bounds hold uniformly over types and stationary policies.

For a matching worst-case lower bound for \textsf{UE} and \textsf{G-UE}, fix a positive type $\theta_i$ and a deterministic stationary policy that gives the good for free at one non-opt-out action and assigns $(0,0)$ to every other action. Both algorithms select each of the $K-1$ other actions $N_T$ times during exploration, incurring time-average external regret of at least $(K-1)\theta_i N_T/T=\Omega(T^{-1/3}(\log T)^{1/3})$.

\section{Proofs of the extraction results}
\label{app:extraction}

\subsection{Uniform full extraction from no-ER algorithms (Theorem~\ref{thm:impossibility_noER})}
\label{app:noer_extraction}

\begin{proof}[Proof of Theorem~\ref{thm:impossibility_noER}]
Define the nonnegative time-average regret bound
\begin{align}
\epsilon_T^{\mathrm{ER}}
\coloneqq\max\left\{0,\max_i\sup_{\sigma_P^T}
\mathrm{ER}(\mathcal L;T,\theta_i,\sigma_P^T)\right\}.
\label{eq:uniform_average_er_bound}
\end{align}
Then $\epsilon_T^{\mathrm{ER}}\to0$ by Definition~\ref{df:no_ER}. Let
\[
R_T^{\mathrm{cum}}\coloneqq T\epsilon_T^{\mathrm{ER}}=o(T).
\]

Suppose that the agent's type is $\theta_i$.
Fix $\epsilon>0$ and a purchase action $b\neq a_0$. For $k=2,\ldots,n$, define
\[
q_k=\frac{\theta_{k-1}+\theta_k}{2},\qquad
d_k=\frac{\theta_k-\theta_{k-1}}{2}.
\]
Choose $\delta>0$ with $p_{\mathrm{max}}(n-1)\delta<\epsilon/4$. Starting from $S_{n+1}=0$, define the following quantities recursively for $k=n,n-1,\ldots,2$.
\begin{align}
C_k&=R_T^{\mathrm{cum}}+p_{\mathrm{max}}S_{k+1},&
N_k&=\max\{1,\lceil2C_k/(\delta d_k)\rceil\},&
S_k&=S_{k+1}+N_k.
\label{eq:app_probe_lengths}
\end{align}
The recursion has finitely many steps with fixed coefficients, so $S_2=O(R_T^{\mathrm{cum}}+1)=o(T)$. For large $T$, run blocks $n,n-1,\ldots,2$, offering $(1,q_k)$ at $b$ for $N_k$ periods in block $k$, with zero outcomes at all other actions. Write $Z_k$ for the number of purchases in block $k$ and $H_k=\{Z_k\geq N_k/2\}$.

To bound classification errors, consider an auxiliary policy of the principal that runs the same blocks through block $k$ and then gives every action $(0,0)$ until horizon $T$. The distribution of actions and observed outcomes through the end of testing block $k$ is unchanged. For $i\geq k$, comparison with the fixed action $b$ gives
\[
d_k\E_i[N_k-Z_k]\leq R_T^{\mathrm{cum}}+p_{\mathrm{max}}S_{k+1}=C_k
\]
because each missed purchase in block $k$ costs the agent at least $d_k$ relative to choosing $b$, earlier blocks contribute at least $-p_{\mathrm{max}}S_{k+1}$ to that utility difference, and the continuation contributes zero. For $i<k$, all prices posted at $b$ through the end of block $k$ exceed $\theta_i$, so comparison with $a_0$ gives $d_k\E_i Z_k\leq R_T^{\mathrm{cum}}\leq C_k$.
Markov's inequality yields
\[
\Pr_i(H_k^c)\leq\delta\quad(i\geq k),\qquad
\Pr_i(H_k)\leq\delta\quad(i<k).
\]
Set $\widehat i=\max(\{1\}\cup\{k:H_k\})$ after testing. A union bound gives $\Pr_i(\widehat i\neq i)\leq(n-1)\delta$.

Let $\theta_+=\min(\Theta\setminus\{0\})$ and choose $0<\Delta<\min\{\epsilon/4,\theta_+/2\}$. In the remaining $T-S_2$ periods, if the estimated type is positive, assign $(1,\theta_{\widehat i}-\Delta)$ to $b$ and $(0,0)$ to every other action. If the estimated type is zero, assign $(0,0)$ to every action. This proposed extraction policy is feasible using only observed actions.

Fix a true type $\theta_i>0$. A second auxiliary policy runs the same tests and then always posts $\theta_i-\Delta$ at $b$, regardless of the test results. We use this type-dependent policy only in the proof. Let $M_i^{\mathrm{aux}}$ count the periods after testing in which the learner does not choose $b$ under this policy. Comparison with $b$ gives
\begin{align}
\Delta\E_i M_i^{\mathrm{aux}}\leq R_T^{\mathrm{cum}}+p_{\mathrm{max}}S_2.
\label{eq:app_aux_missed_trades}
\end{align}
Run the auxiliary policy and the extraction policy against two copies of the learner, both with true type $\theta_i$.
Use the same realization of the learner's internal random choices in both runs, so that identical observed histories lead to identical actions.
The two runs agree throughout testing and, on $E_i=\{\widehat i=i\}$, throughout the continuation. Let $M_i$ count the periods after testing in which the learner does not choose $b$ under the extraction policy. Then $\E_i[M_i\mathbf1_{E_i}]\leq\E_i M_i^{\mathrm{aux}}$. This step uses an unconditional regret bound under an auxiliary policy, not a regret bound conditional on successful identification.

For a lower bound on total payments, count only the payments made after testing on paths where the estimated type is correct. This gives
\[
\sum_{t=1}^T p_t(a_t) \geq (\theta_i - \Delta)(T - S_2 - M_i)\mathbf1_{E_i}.
\]
Thus, we have
\begin{align}
\bar p_{i,T}
\geq(\theta_i-\Delta)\left[
1-\frac{S_2}{T}-(n-1)\delta
-\frac{R_T^{\mathrm{cum}}+p_{\mathrm{max}}S_2}{\Delta T}\right].
\label{eq:app_noer_revenue}
\end{align}
It follows that $\theta_i-\bar p_{i,T}\leq\Delta+p_{\mathrm{max}}(n-1)\delta+o(1)<\epsilon$ for all sufficiently large $T$, uniformly over the finitely many positive types. For a zero type, nonnegative payments already give the required bound.
\end{proof}

\subsection{Characterization of the payoffs under no-ER (Corollary~\ref{cor:noer_zero_utility})}
\label{app:noer_zero_utility}

\begin{proof}
The uniform regret bound \eqref{eq:uniform_average_er_bound}, applied to the opt-out action, gives $\bar u_{i,T}\geq-\epsilon_T^{\mathrm{ER}}$ and hence $\bar p_{i,T}\leq\theta_i+\epsilon_T^{\mathrm{ER}}$ under every principal strategy. Together with Theorem~\ref{thm:impossibility_noER}, this yields $V_T(\pi,\mathcal L)\to\sum_i\pi_i\theta_i$.

Along any asymptotically optimal response, the nonnegative deficits $\theta_i+\epsilon_T^{\mathrm{ER}}-\bar p_{i,T}$ have a prior-weighted sum that tends to zero. Full support forces each deficit to vanish, so $\bar p_{i,T}\to\theta_i$. The bounds $-\epsilon_T^{\mathrm{ER}}\leq\bar u_{i,T}\leq\theta_i-\bar p_{i,T}$ give $\bar u_{i,T}\to0$, and the payoff identity gives $\bar x_{i,T}\to1$ for positive types.
\end{proof}

\subsection{Extraction from UE and SE (Proposition~\ref{prop:ue_se_extraction})}
\label{app:ue_se_extraction}

\begin{proof}[Proof of Proposition~\ref{prop:ue_se_extraction}]
Let $N_T=\lceil T^{2/3}(\log T)^{1/3}\rceil$ and $H_T=KN_T=o(T)$. For the first $H_T$ periods, offer free trade when the agent chooses $b$ and assign $(0,0)$ to every other action. Then raise the price charged for choosing $b$ to $p_{\mathrm{max}}$, leaving the other outcomes at zero. For each positive type, the observed utility of $b$ during exploration is exactly $\theta_i>0$, so UE commits to $b$ at the end of exploration.

For SE as specified in Appendix~\ref{app:stationary_rates}, the empirical mean utility of each action during the free-trade phase equals that action's utility, which is constant over the phase. The best action $b$ is never eliminated, and all other actions are eliminated once $2r_s<\theta_+$. This takes $O(\log T)$ complete rounds, hence $O(K\log T)<H_T$ periods for all sufficiently large $T$. Afterwards $b$ is the only active action and remains so when its price rises. Therefore, for either algorithm and every positive type,
\[
\bar p_{i,T}=p_{\mathrm{max}}\frac{T-H_T}{T}
\longrightarrow p_{\mathrm{max}}>\theta_i.
\]
For a zero type, nonnegative payments imply $\liminf_{T\to\infty}\bar p_{i,T}\geq0=\theta_i$. The same deterministic policy establishes these conclusions for both algorithms and all types, without using the prior.
\end{proof}

\subsection{Extraction from UCB (Proposition~\ref{prop:ucb_full_extraction})}
\label{app:ucb_extraction}

\begin{proof}[Proof of Proposition~\ref{prop:ucb_full_extraction}]
\emph{Identification.} For the first $H_T=\lceil\sqrt T\rceil$ periods, charge zero for every action and allocate the good if and only if the agent chooses $b$. Fix a positive type and let $d_i \coloneqq \theta_i/\bar B$.
During identification, the empirical utility is $\theta_i$ at $b$ and zero at every other action. If UCB selects $a\neq b$ after initialization at time $t\leq H_T$, its selection rule implies
\[
\bar B\sqrt{\frac{2\log T}{N_a(t-1)}}
\geq\theta_i+\bar B\sqrt{\frac{2\log T}{N_b(t-1)}}
\geq\theta_i.
\]
Thus $\sqrt{2\log T/N_a(t-1)}\geq d_i$. Applying this at each action's last such selection gives
\[
N_a(H_T)\leq2\log T/d_i^2+1\quad(a\neq b),\qquad
N_b(H_T)=H_T-O(\log T).
\]
At the last selection of $b$ during the first $H_T$ periods, its preceding count is $N_b(H_T)-1$. Applying UCB's selection rule to compare $b$ with $a_0$ gives
\[
\frac{2}{(d_i+\sqrt{2\log T/(N_b(H_T)-1)})^2}
\leq\frac{N_{a_0}(H_T)}{\log T}
\leq\frac{2}{d_i^2}+\frac{1}{\log T}.
\]
This proves \eqref{eq:ucb_type_identification}. If the type is zero, all empirical utilities are zero, so UCB always selects a least-sampled action and the counts differ by at most one. Classify the agent as type zero if $0\in\Theta$ and $N_{a_0}(H_T)>H_T/(2K)$. Otherwise, choose the positive type whose target $2/d_i^2$ is closest to $N_{a_0}(H_T)/\log T$. The targets are finitely many and distinct, so this estimator is exact for all sufficiently large horizons.

\emph{Extraction.} Put $\epsilon_T=\min\{\theta_+/2,T^{-1/8}\}$, where $\theta_+=\min(\Theta\setminus\{0\})$. If the estimated type $\widehat\theta$ is positive, charge $\widehat\theta-\epsilon_T$ for trade when the agent chooses $b$ and assign $(0,0)$ to the other actions. If the estimated type is zero, assign $(0,0)$ to every action. Once identification is exact, every trading observation for a positive type has utility at least $\epsilon_T$, including the initial free observations. At any later selection of $a\neq b$, UCB's selection rule therefore requires $\bar B\sqrt{2\log T/N_a(t-1)}\geq\epsilon_T$. The last-selection argument yields
\[
T-N_b(T)\leq(K-1)
\left(1+\frac{2\bar B^2\log T}{\epsilon_T^2}\right)
=O(T^{1/4}\log T)=o(T).
\]
The principal charges $\theta_i-\epsilon_T$ on at least $N_b(T)-H_T$ trades, so time-average payments converge to $\theta_i$. On every path, time-average utility lies between zero and $\theta_i H_T/T+\epsilon_T$ and therefore vanishes. Time-average allocation tends to one.
For all sufficiently large $T$, a zero type makes no payments and receives zero utility under this policy.

\emph{An upper bound for every principal policy.} For arbitrary feasible play, fix $a\neq a_0$. If the agent last selects $a$ after initialization, let $s\geq1$ count its earlier selections and let $S_s$ be the sum of the utilities from those selections. Since the opt-out action always yields zero utility, UCB's selection rule implies $S_s/s+\bar B\sqrt{2\log T/s}\geq0$, and thus $S_s\geq-\bar B\sqrt{2s\log T}$. The last observation contributes at least $-p_{\mathrm{max}}$. If $a$ was selected only once, this latter bound alone suffices.

The count before the last selection of $a$ is at most $N_a(T)$, its total number of selections. Thus the periods in which $a\neq a_0$ is selected contribute at least $-\bar B\sqrt{2N_a(T)\log T}-p_{\mathrm{max}}$ to total utility. This bound also holds if $a$ is never selected, in which case its contribution is zero. The opt-out action contributes zero utility, so it suffices to sum over the $K-1$ other actions. By Cauchy--Schwarz,
\[
\sum_{a\neq a_0}\sqrt{N_a(T)}
\leq\sqrt{(K-1)\sum_{a\neq a_0}N_a(T)}
\leq\sqrt{(K-1)T},
\]
where the last inequality holds because exactly one action is selected per period. Summing the per-action bounds and using this inequality gives
\[
\sum_{t=1}^T u_t
\geq-\bar B\sqrt{2(K-1)T\log T}-(K-1)p_{\mathrm{max}}.
\]
Thus $\bar u_{i,T}\geq-e_T$ and $\bar p_{i,T}\leq\theta_i+e_T$ for a common $e_T=o(1)$, uniformly over the principal strategies. Together with the extraction policy, this gives $V_T(\pi,\textsf{UCB})\to\sum_i\pi_i\theta_i$. The same argument as in the proof of Corollary~\ref{cor:noer_zero_utility} then gives the stated payment, utility, and positive-type allocation limits along every asymptotically optimal response.
\end{proof}

\subsection{Failure of rent protection (Corollary~\ref{cor:standard_algorithms_fail_no_dre})}
\label{app:standard_no_dre}

\begin{proof}[Proof of Corollary~\ref{cor:standard_algorithms_fail_no_dre}]
Choose a full-support prior that places enough weight on the smallest positive type $\theta_+$ to make $\theta_+$ the unique optimal static price. Then type $\theta_n$ has static screening rent $u^*(\theta_n,\pi)=\theta_n-\theta_+>0$, whereas Corollary~\ref{cor:noer_zero_utility} and Proposition~\ref{prop:ucb_full_extraction} give zero limiting utility for $\theta_n$ along every asymptotically optimal response to no-ER and UCB, respectively. Thus these algorithms do not satisfy no dynamic rent extraction.

For UE and SE, Proposition~\ref{prop:ue_se_extraction} and nonnegative payments at a possible zero type give
\[
\liminf_T V_T(\pi,\mathcal L)
\geq p_{\mathrm{max}}\sum_{i:\theta_i>0}\pi_i
>\sum_i\pi_i\theta_i\geq\operatorname{Rev}^*(\pi).
\]
This contradicts the revenue ceiling required by no dynamic rent extraction in Proposition~\ref{prop:no_dre_revenue_ceiling}.
\end{proof}

\section{The static screening benchmark}
\label{app:static_screening}

\subsection{Posted-price characterization (Proposition~\ref{prop:static_screening_characterization})}
\label{app:proof_static_screening}

\begin{proof}[Proof of Proposition~\ref{prop:static_screening_characterization}]
Posted-price optimality follows from \citet[Section~5, Proposition~2]{Skreta2006}. To characterize the entire optimal set and each type's static screening rent, we use the finite-type decomposition below.

Incentive compatibility implies that allocations are nondecreasing in type.  At a revenue optimum, the individual rationality constraint for the lowest type and all adjacent downward incentive constraints bind. If any were slack, the principal could lower some utilities and raise revenue.  Set $x_0=0$ and $\lambda_k=x_k-x_{k-1}$.
The finite-type version of the cutoff decomposition in \citet[proof of Proposition~2, Step~2]{Skreta2005TechnicalAppendix} gives
\[
p_i=\sum_{k\leq i}\theta_k\lambda_k,
\qquad
U_i(D)=\sum_{k=1}^n\lambda_k(\theta_i-\theta_k)_+,
\qquad
R_\pi(D)=\sum_{k=1}^n\lambda_k R_\pi(D^k).
\]
Here $\lambda_k\geq0$ and $\sum_k\lambda_k=x_n\leq1$. Because the type set contains a positive type and the prior has full support, the maximum posted-price revenue is positive. Revenue is therefore maximized exactly when $\sum_k\lambda_k=1$ and positive weight is placed only on revenue-maximizing posted prices. Among these prices, the highest one, $p^*(\pi)$, minimizes every type's utility simultaneously and thus yields the static screening rents in \eqref{eq:explicit_static_rent}.
\end{proof}

\section{Robust learning (Theorem~\ref{thm:robust_algorithms})}
\label{app:robust_proofs}

We prove Theorem~\ref{thm:robust_algorithms} by showing how each algorithm's empirical menu induces a static direct mechanism. The same argument will also apply to unknown i.i.d.\ values. The stationary bounds were established in Appendix~\ref{app:stationary_rates}.

\subsection{From empirical menus to protected rents}
\label{app:menu_transfer}

Let $\mathcal Z=([0,1]\times[0,p_{\mathrm{max}}])^n$ and, for $\kappa\geq0$, define the $\kappa$-approximately IC and IR static mechanisms
\[
\mathcal F_\kappa=\left\{D\in\mathcal Z:
U_i(D)\geq-\kappa,\quad
U_i(D)\geq\theta_iX_j-P_j-\kappa
\ \text{for every }i,j\right\}.
\]
Here $D=(X_i,P_i)_i$ and $U_i(D)=\theta_iX_i-P_i$, so $\mathcal F_0=\mathcal F$.

\begin{lem}
\label{lem:app_menu_transfer}
Suppose a consistent algorithm has the following representation against every principal's strategy.
There are random variables $D_T=(X_{i,T},P_{i,T})_i\in\mathcal F_{\kappa_T}$ and $\Gamma_T=(\Gamma_{i,T})_i\in[0,1]^n$ such that
\begin{align}
\left|\bar x_{i,T}-\E[\Gamma_{i,T}X_{i,T}]\right|\leq e_T,\qquad
\left|\bar p_{i,T}-\E[\Gamma_{i,T}P_{i,T}]\right|\leq e_T.
\label{eq:app_menu_accounting}
\end{align}
The deterministic bounds $\kappa_T$ and $e_T$ tend to zero and do not depend on the type or the principal's strategy. Then the algorithm is asymptotically individually rational, satisfies no dynamic rent extraction, and has $V_T(\pi,\mathcal L)\to\operatorname{Rev}^*(\pi)$ for every full-support prior.
\end{lem}

\begin{proof}
Utility accounting follows from the payoff identity:
\begin{align}
\bar u_{i,T}
=\E[\Gamma_{i,T}U_i(D_T)]+O(e_T)
\geq-\kappa_T-O(e_T).
\label{eq:app_utility_accounting}
\end{align}
This establishes asymptotic individual rationality regardless of the principal's strategy.

Fix $\pi$. The laws of $(D_T,\Gamma_T)$ lie on the compact space $\mathcal Z\times[0,1]^n$.
For any sequence of horizons tending to infinity and any corresponding sequence of strategies for the principal, the joint laws have a weakly convergent subsequence. Write $(D,\Gamma)$ for a random pair with the limiting law.
For every integer $m\geq1$, the set $\mathcal F_{1/m}$ is closed because the functions defining the feasibility constraints are continuous. Since $\kappa_T\to0$, $D_T\in\mathcal F_{1/m}$ almost surely for all sufficiently large $T$, so weak convergence gives $D\in\mathcal F_{1/m}$ almost surely. Intersecting these sets over $m$ gives $D\in\mathcal F$ almost surely.
Along this subsequence, we have
\begin{align}
\lim_T\operatorname{Rev}_T
=\E\sum_i\pi_i\Gamma_iP_i
\leq\E R_\pi(D)\leq\operatorname{Rev}^*(\pi).
\label{eq:app_menu_revenue_limit}
\end{align}
By \eqref{eq:app_menu_accounting}, type $i$'s expected time-average payment $\bar p_{i,T}$ differs by at most $e_T$ from the expected weighted payment $\E[\Gamma_{i,T}P_{i,T}]$ in the static representation. Averaging over types with weights $\pi_i$ gives the same error bound for revenue. Since $e_T\to0$ and the map $(D,\Gamma)\mapsto\sum_i\pi_i\Gamma_iP_i$ is bounded and continuous, weak convergence gives the equality. The first inequality follows from $\Gamma_i\leq1$ and $P_i\geq0$. The second follows from $D\in\mathcal F$ almost surely.

Since $V_T$ is a supremum, we can choose horizons tending to infinity and corresponding principal strategies whose revenues converge to $\limsup_T V_T$. Applying \eqref{eq:app_menu_revenue_limit} along a subsequence on which the joint laws of $(D_T,\Gamma_T)$ converge weakly gives $\limsup_T V_T\leq\operatorname{Rev}^*(\pi)$. Conversely, Lemma~\ref{lem:stationary_cutoff_outcomes} applied to an optimal positive cutoff gives $\liminf_TV_T\geq\operatorname{Rev}^*(\pi)$. Thus we have $V_T(\pi,\mathcal L)\to\operatorname{Rev}^*(\pi)$.

Along any asymptotically optimal response, both inequalities in \eqref{eq:app_menu_revenue_limit} must be equalities for every subsequential limit. Equivalently,
\[
\E\left[\operatorname{Rev}^*(\pi)-R_\pi(D)
+\sum_i\pi_i(1-\Gamma_i)P_i\right]=0.
\]
Every summand is nonnegative. Full support therefore implies $D\in\mathcal D^*(\pi)$ and $(1-\Gamma_i)P_i=0$ almost surely for every $i$.
Let $c$ be the lowest revenue-maximizing posted price, which is positive since $\operatorname{Rev}^*(\pi)>0$. By Proposition~\ref{prop:static_screening_characterization}, every optimal static mechanism satisfies $P_i\geq cX_i$. It follows that $(1-\Gamma_i)X_i=0$, and consequently
\[
\Gamma_iU_i(D)=U_i(D)\geq u^*(\theta_i,\pi)
\quad\text{almost surely}.
\]
Fix a type \(i\). Choose horizons \(T_m\) along which \(\bar u_{i,T_m}\) converges to \(\liminf_T\bar u_{i,T}\), and then extract a further subsequence along which the joint laws of \((D_{T_m},\Gamma_{T_m})\) converge weakly.
The map $(D,\Gamma)\mapsto\Gamma_i U_i(D)$ is bounded and continuous, so weak convergence and \eqref{eq:app_utility_accounting} yield $\liminf_T\bar u_{i,T}\geq u^*(\theta_i,\pi)$, the static screening rent lower bound required by no dynamic rent extraction.
\end{proof}

\subsection{Near-optimal menu choices under elimination}
\label{app:elimination_menu}

\begin{lem}
\label{lem:app_replacement}
Fix numbers $(v(a))_a$ for the $K$ actions and a completed elimination round $r$.
Suppose that, through round $r$, both confidence endpoints of every active action $a$ differ from $v(a)$ by at most $h$. Under the simultaneous elimination rule \eqref{eq:ue_se_elimination_rule}, every survivor $a$ of round $r$ satisfies
\[
v(a)\geq\max_b v(b)-2Kh.
\]
\end{lem}

\begin{proof}
Whenever an action $b$ is eliminated, choose an action $c$ with maximal lower confidence endpoint in that round. It survives, and the endpoint bounds imply $v(c)\geq v(b)-2h$. Start with a maximizer of $v$ and replace it this way whenever it is eliminated. Each replacement survives the current round, so any later replacement occurs in a strictly later round.
By the end of round $r$, at most $K-1$ replacements occur, leaving an action that survives round $r$ with value at least $\max_b v(b)-2(K-1)h$. Every other survivor's upper confidence endpoint is at least the lower confidence endpoint of the action obtained by this procedure. Its value is therefore at most $2h$ below that action's value.
\end{proof}

For \textsf{G-SE}, write $v_i(a)=\theta_i\bar x_a-\bar p_a$ for type $i$'s utility from action $a$ in the empirical menu. Before the first failure of a consistency test, interval overlap gives
\[
|\widehat x_{a,m_T+s}-\bar x_a|\leq2\rho_{m_T},\qquad
|\widehat p_{a,m_T+s}-\bar p_a|\leq2\rho_{m_T}.
\]
Thus the utility estimate differs from $v_i(a)$ by at most $2(\bar\theta+1)\rho_{m_T}$.
The confidence interval has half-width at most $\rho_{m_T}$, so Lemma~\ref{lem:app_replacement} applies with $h=[2(\bar\theta+1)+1]\rho_{m_T}$. Every action selected after the first completed elimination round and before the period in which a consistency test first fails is therefore $\kappa_T$-optimal in the empirical menu, where
\begin{align}
\kappa_T=2K[2(\bar\theta+1)+1]\rho_{m_T}.
\label{eq:app_gse_menu_error}
\end{align}
The first round is excluded because it samples all actions before any elimination. Its at most $K$ observations are accounted for below. For \textsf{G-UE}, the selected action maximizes utility in the empirical menu, so $\kappa_T=0$.

\subsection{Allocation and payment bounds}
\label{app:guarded_representation}

Let $\tau$ be the first Phase-2 date on which either the allocation or the payment consistency test fails, with $\tau=T+1$ if neither test ever fails. If $\tau\leq T$, the agent observes the outcome of the action taken in period $\tau$ and opts out in every subsequent period. Both allocation and payment are zero after period $\tau$. First consider all Phase-2 observations strictly before $\tau$.

For \textsf{G-UE}, if there are $r$ such observations of the repeated action and $\bar y$ is its exploration mean, interval overlap implies, for $y\in\{x,p\}$,
\begin{align}
\left|\sum_{s=1}^r(y_s-\bar y)\right|
\leq r\rho_{N_T}+C_0\sqrt{r\log T}.
\label{eq:app_gue_accounting}
\end{align}
For \textsf{G-SE}, if action $a$ has $s_a$ such observations, its all-sample mean satisfies
\begin{align}
\left|\sum_{s=1}^{s_a}(y_{a,s}-\bar y_a)\right|
&=(m_T+s_a)|\widehat y_{a,m_T+s_a}-\bar y_a|
\leq2(m_T+s_a)\rho_{m_T}.
\label{eq:app_gse_accounting}
\end{align}
Both bounds are pathwise, including against an adaptive principal.

Let $q_T=N_T$ for \textsf{G-UE} and $q_T=m_T$ for \textsf{G-SE}, and let $Z_T=(a_t,x_t,p_t)_{t\leq Kq_T}$ be the principal's complete Phase-1 history. The distribution of $Z_T$ is the same for every type. The action schedule is deterministic and type independent, and for a fixed strategy the principal draws each mechanism from a distribution determined by her observed history. For a Phase-1 history $z$, let $(\bar x_a(z),\bar p_a(z))$ denote the average allocation and payment observed when action $a$ was selected.

For each type $i$, let $N_{i,a}$ count selections of action $a$ in Phase~2 strictly before $\tau$, excluding the first elimination round for \textsf{G-SE}. Given $Z_T=z$, set
\begin{align}
c_{i,a}(z)&=\E_i[N_{i,a}\mid Z_T=z],&
J_i(z)&=\sum_a c_{i,a}(z),&
\Gamma_{i,T}(z)&=J_i(z)/T.
\label{eq:app_conditional_counts}
\end{align}
For sufficiently large $T$, at least one Phase-2 selection is counted on every path. For \textsf{G-UE}, $\rho_1>B_0$ ensures that the first Phase-2 test passes. For \textsf{G-SE}, adding $s$ Phase-2 observations of an action can change either empirical mean by at most $sB_0/(m_T+s)$. This is at most $\rho_{m_T+s}$ for $s\in\{1,2\}$, so each action's tests pass at its first two Phase-2 selections whenever they occur. The first selection after the initial elimination round is the second Phase-2 selection of that action, so it is counted.
Thus $J_i(z)\geq1$ almost surely. Define
\begin{align}
(X_{i,T}(z),P_{i,T}(z))
=\sum_a\frac{c_{i,a}(z)}{J_i(z)}
       (\bar x_a(z),\bar p_a(z)).
\label{eq:app_menu_rows}
\end{align}
The pairs $(X_{i,T},P_{i,T})$ form the static direct mechanism $D_T$ by assigning each type a mixture of entries from the common empirical menu. These are analytical objects.

By the preceding subsection, each action with positive weight in \eqref{eq:app_menu_rows} is $\kappa_T$-optimal in the empirical menu. Consequently,
\[
U_i(D_T)\geq\max_a v_i(a)-\kappa_T,\qquad
\theta_iX_{j,T}-P_{j,T}\leq\max_a v_i(a).
\]
The empirical menu contains the opt-out outcome, so $D_T\in\mathcal F_{\kappa_T}$.

To verify \eqref{eq:app_menu_accounting}, compare the actual allocation and payment totals with the empirical-menu totals for the counted Phase-2 selections. For \textsf{G-UE}, \eqref{eq:app_gue_accounting}, the $KN_T$ exploration dates, and the observation at date $\tau$, if any, give
\[
e_T=O\left(\frac{KN_T+1}{T}
+\rho_{N_T}+\sqrt{\frac{\log T}{T}}\right)
\]
in \eqref{eq:app_menu_accounting}. For \textsf{G-SE}, sum \eqref{eq:app_gse_accounting} over actions. Removing the first elimination round changes the allocation and payment discrepancies by at most $K$ and $Kp_{\mathrm{max}}$, respectively. The comparison omits at most $Km_T+K+1$ periods in total, counting exploration, the first elimination round, and period $\tau$ if a test fails. Thus one may take
\[
e_T=O\left(\frac{Km_T+2K+1}{T}
+\left(1+\frac{Km_T}{T}\right)\rho_{m_T}\right).
\]
By \eqref{eq:app_conditional_counts} and \eqref{eq:app_menu_rows},
\[
\Gamma_{i,T}(z)(X_{i,T}(z),P_{i,T}(z))
=\frac{1}{T}\E_i\left[\sum_a N_{i,a}(\bar x_a(z),\bar p_a(z))\,\middle|\, Z_T=z\right].
\]
Averaging over the common law of $Z_T$ and using the preceding bounds yields \eqref{eq:app_menu_accounting}. For both algorithms, $e_T\to0$ and $\kappa_T\to0$ uniformly over types and principal strategies; here $\kappa_T=0$ for \textsf{G-UE} and is given by \eqref{eq:app_gse_menu_error} for \textsf{G-SE}.

\begin{proof}[Completion of the proof of Theorem~\ref{thm:robust_algorithms}]
Equations~\eqref{eq:app_gue_rate} and \eqref{eq:app_gse_rate} establish consistency of both algorithms. The representation just constructed satisfies Lemma~\ref{lem:app_menu_transfer}; it therefore proves the revenue limit, no dynamic rent extraction, and asymptotic individual rationality. By \eqref{eq:app_gse_rate}, the worst-case external regret of \textsf{G-SE} for stationary policies is $\widetilde O(T^{-1/2})$. This matches the minimax rate in Section~\ref{sec:AOS} up to logarithmic factors, so \textsf{G-SE} is rate-optimal.
\end{proof}

\section{Optimality of rent protection}
\label{app:optimality_proofs}

\subsection{The revenue ceiling (Proposition~\ref{prop:no_dre_revenue_ceiling})}
\label{app:proof_no_dre_revenue_ceiling}

\begin{proof}
Let $\theta_+\coloneqq\min(\Theta\setminus\{0\})$, and write $q_0=q(0)$ when $0\in\Theta$ and $q_0=0$ otherwise.  First consider any full-support prior $q$ for which $\theta_+$ is the unique revenue-maximizing posted price.  Choose, at each horizon, a principal strategy whose revenue is within $1/T$ of $V_T(q,\mathcal L)$.  This sequence is asymptotically optimal.  By Proposition~\ref{prop:static_screening_characterization}, the static screening rent at $q$ is $(\theta_i-\theta_+)_+$ for each type $\theta_i$. No dynamic rent extraction therefore implies
\[
\liminf_T\bar u_{i,T}\geq \theta_i-\theta_+
\quad\text{for every }\theta_i>0,
\qquad
\liminf_T\bar u_{i,T}\geq0
\quad\text{if }\theta_i=0.
\]
Because $\bar x_{i,T}\leq1$ and $\bar p_{i,T}=\theta_i\bar x_{i,T}-\bar u_{i,T}$, it follows that
\[
\limsup_T\bar p_{i,T}\leq\theta_+
\quad\text{for every }\theta_i>0,
\qquad
\limsup_T\bar p_{i,T}\leq0
\quad\text{if }\theta_i=0.
\]
The type set is finite, so the $1/T$-optimality of the selected responses gives
\begin{align}
\limsup_TV_T(q,\mathcal L)\leq\theta_+(1-q_0).
\label{eq:lowest_type_revenue_cap}
\end{align}

Now fix an arbitrary full-support prior $\pi$, and let $\pi_0=\pi(0)$ if $0\in\Theta$ and $\pi_0=0$ otherwise.  Since posting $\theta_+$ yields revenue $\theta_+(1-\pi_0)$, we have $\operatorname{Rev}^*(\pi)\geq\theta_+(1-\pi_0)$.  Define
\[
\alpha_*\coloneqq
\frac{\theta_+}{\operatorname{Rev}^*(\pi)+\theta_+\pi_0}
\in(0,1].
\]
For any fixed $\alpha\in(0,\alpha_*)$, consider the full-support prior
\[
q_\alpha\coloneqq\alpha\pi+(1-\alpha)\delta_{\theta_+}.
\]
At this prior, price $\theta_+$ earns $\theta_+(1-\alpha\pi_0)$.
Every higher price $p$ earns
\[
\alpha p\Pr_\pi(\theta\geq p)
\leq\alpha\operatorname{Rev}^*(\pi)
<\theta_+(1-\alpha\pi_0),
\]
where the last inequality follows from $\alpha < \alpha_*$.
Thus $\theta_+$ is the unique optimal price at $q_\alpha$, and we have
\[
\limsup_TV_T(q_\alpha,\mathcal L)\leq\theta_+(1-\alpha\pi_0).
\]

We have
\[
\operatorname{Rev}_T(\sigma_P^T;q_\alpha,\mathcal L)
=\alpha\operatorname{Rev}_T(\sigma_P^T;\pi,\mathcal L)
 +(1-\alpha)\bar p_{+,T}
\geq\alpha\operatorname{Rev}_T(\sigma_P^T;\pi,\mathcal L),
\]
where $\bar p_{+,T}$ is the payment conditional on type $\theta_+$. Taking the supremum over the principal's strategies at each horizon and then taking the limit superior gives
\[
\limsup_TV_T(\pi,\mathcal L)
\leq\frac{\theta_+(1-\alpha\pi_0)}{\alpha}.
\]
Letting $\alpha\uparrow\alpha_*$ yields $\limsup_TV_T(\pi,\mathcal L)\leq\operatorname{Rev}^*(\pi)$. This concludes the proof.
\end{proof}

\subsection{Stationary cutoff outcomes (Lemma~\ref{lem:stationary_cutoff_outcomes})}
\label{app:proof_stationary_cutoff_outcomes}

\begin{proof}
Let
\[
e_T\coloneqq
\sup_{\theta_i\in\Theta,\,\mu\in\Delta(\mathcal M)}
\operatorname{ER}(\mathcal L(\theta_i,T);T,\theta_i,\mu),
\]
so $e_T\to0$ by consistency.  Define
\[
g_p\coloneqq
\min\bigl(\{p\}\cup\{p-\theta_i:\theta_i<p\}\bigr)>0,
\qquad
\delta_T\coloneqq
\min\{g_p/2,\sqrt{e_T}+T^{-1/2}\}.
\]
Then $\delta_T\to0$ and $e_T/\delta_T\to0$. Fix a non-opt-out action $b$ and post $(1,p-\delta_T)$ at $b$ and $(0,0)$ elsewhere in every period. Let $N_{i,T}$ count type $\theta_i$'s choices of $b$. When $\theta_i\geq p$, each missed trade costs the agent at least $\delta_T$ in utility. When $\theta_i<p$, each trade costs the agent at least $g_p/2$ in utility relative to opting out. Uniform regret therefore gives
\[
1-\frac{\E N_{i,T}}{T}\leq\frac{e_T}{\delta_T}
\quad(\theta_i\geq p),
\qquad
\frac{\E N_{i,T}}{T}\leq\frac{2e_T}{g_p}
\quad(\theta_i<p).
\]
Both bounds vanish, so $\E N_{i,T}/T\to\mathbf1\{\theta_i\geq p\}$, including for a zero type when present. The expected time-average allocation and payment are $\E N_{i,T}/T$ and $(p-\delta_T)\E N_{i,T}/T$, respectively. Their limits and the payoff identity give \eqref{eq:stationary_cutoff_outcomes}.
\end{proof}

\subsection{Exact rent guarantees (Theorem~\ref{thm:sharpness_rent_protection})}
\label{app:proof_sharpness_rent_protection}

\begin{proof}
Apply Lemma~\ref{lem:stationary_cutoff_outcomes} to any revenue-maximizing price.  The resulting stationary strategy sequence has limiting revenue $\operatorname{Rev}^*(\pi)$, so it gives
\[
\liminf_TV_T(\pi,\mathcal L)\geq\operatorname{Rev}^*(\pi).
\]
Proposition~\ref{prop:no_dre_revenue_ceiling} gives the reverse upper bound. Apply Lemma~\ref{lem:stationary_cutoff_outcomes} specifically to $p^*(\pi)$. The revenue from the stationary strategy sequence at this price converges to $\operatorname{Rev}^*(\pi)$, which is also the limit of $V_T(\pi,\mathcal L)$. Thus the sequence is asymptotically optimal.  Under this sequence every type's expected utility converges to his static screening rent $u^*(\theta_i,\pi)$ by \eqref{eq:explicit_static_rent}. This single sequence gives $g_i(\pi,\mathcal L)\leq u^*(\theta_i,\pi)$ for every type simultaneously, while no dynamic rent extraction gives the reverse inequality.
\end{proof}

\subsection{Pareto optimality across \texorpdfstring{type--prior pairs}{type--prior pairs} (Corollary~\ref{cor:pareto_optimal_rent_protection})}
\label{app:proof_pareto_optimal_rent_protection}

\begin{proof}
By Theorem~\ref{thm:robust_algorithms}, each guarded algorithm $\mathcal L^*\in\{\textsf{G-UE},\textsf{G-SE}\}$ is consistent and satisfies no dynamic rent extraction. Theorem~\ref{thm:sharpness_rent_protection} therefore gives $g_i(\pi,\mathcal L^*)=u^*(\theta_i,\pi)$ for every type and full-support prior.  Any consistent algorithm $\widetilde{\mathcal L}$ that weakly dominates these guarantees satisfies no dynamic rent extraction. Applying Theorem~\ref{thm:sharpness_rent_protection} to that algorithm gives $g_i(\pi,\widetilde{\mathcal L})=u^*(\theta_i,\pi)$ everywhere, ruling out a strict improvement. Thus each guarded algorithm is Pareto optimal among consistent algorithms with respect to its rent guarantees across type--prior pairs.
\end{proof}

\section{Unknown i.i.d.\ values (Theorem~\ref{thm:unknown_iid_extension})}
\label{app:unknown_iid_proofs}

\subsection{Noise at adaptive sample counts}

For action $a$, let $\tau_{a,s}$ be the period of its $s$th selection, or $\infty$ if the agent selects $a$ fewer than $s$ times. Define
\[
M_{i,a,s}=\sum_{j=1}^s
x_{\tau_{a,j}}(a)(\theta_{\tau_{a,j}}-\mu_i)
\quad\text{when }\tau_{a,s}\leq T.
\]

Fix a mean $\mu_i$ and an action $a$. Write $I_t=\mathbf1\{a_t=a\}$ and let
\[
S_t=\sum_{u=1}^t I_u x_u(a)(\theta_u-\mu_i),
\qquad N_t=\sum_{u=1}^t I_u .
\]
Conditional on the full past and the current action and mechanism, the increment of $S_t$ has mean zero. It is zero if $I_t=0$ and lies in an interval of length $B$ if $I_t=1$. By Hoeffding's lemma and iterated expectation,
\[
W_t=\exp\{\lambda S_t-\lambda^2B^2 N_t/8\},\qquad W_0=1,
\]
is a nonnegative supermartingale with respect to the filtration generated by all values, actions, and mechanisms through period $t$. At the bounded stopping time $\tau_{a,s}\wedge T$, we have $\E_i[W_{\tau_{a,s}\wedge T}]\leq1$. On $\{\tau_{a,s}\leq T,\ M_{i,a,s}>s\rho_s\}$, the stopped count is $s$ and the stopped sum exceeds $s\rho_s$. Taking $\lambda=4\rho_s/B^2$ bounds this event's probability by $\exp(-2s\rho_s^2/B^2)$. Applying the same argument to $-S_t$ gives
\[
\Pr_i\{\tau_{a,s}\leq T,\ |M_{i,a,s}|>s\rho_s\}
\leq 2\exp(-2s\rho_s^2/B^2)\leq 2T^{-4},
\]
where the last inequality uses $C_{\mathrm{iid}}\geq\sqrt{2}B$. A union bound over actions and sample counts gives the following result.

\begin{lem}
\label{lem:app_iid_noise}
For every mean and principal strategy,
\[
\Pr_i\left\{|M_{i,a,s}|\leq s\rho_s
\text{ for all }a\text{ and every }s\text{ such that }\tau_{a,s}\leq T\right\}
\geq1-2KT^{-3}.
\]
\end{lem}

\subsection{Learning and rent protection}

Under a stationary policy $\nu$, the sampling construction in Appendix~\ref{app:stationary_rates} applies to allocations, payments, and realized rewards. These lie in $[0,1]$, $[0,p_{\mathrm{max}}]$, and $[-p_{\mathrm{max}},B]$, respectively. For each empirical mean based on $s$ observations, Hoeffding's inequality bounds the probability of a deviation greater than $\rho_s$ from its stationary expectation by $2T^{-4}$.
For each action, we consider averages of its first $s$ observations, which cover both algorithms' exploration estimates and \textsf{G-SE}'s estimates pooling both phases. For \textsf{G-UE}, we also consider averages of its first $s$ Phase-2 observations, which start after the $N_T$ exploration observations of that action. A union bound over both windows, all three quantities, all actions, and all sample counts gives a failure probability of at most $12KT^{-3}$.
The same regret calculation gives
\begin{align}
\operatorname{ER}(\textsf{G-UE}^{\mathrm{iid}}(T);T,\mu_i,\nu)
&=O\left(\frac{KN_T}{T}+\rho_{N_T}+KT^{-3}\right),
\notag\\
\operatorname{ER}(\textsf{G-SE}^{\mathrm{iid}}(T);T,\mu_i,\nu)
&=O\left(\frac{Km_T}{T}+\sqrt{\frac{K\log T}{T}}+KT^{-3}\right).
\label{eq:app_iid_stationary_rates}
\end{align}

Now fix an arbitrary adaptive principal. Let $Z_T$ be the principal's complete Phase-1 history, as defined in Appendix~\ref{app:guarded_representation}. Set $v_i(a)=\mu_i\bar x_a-\bar p_a$, the expected utility of action $a$ in the Phase-1 empirical menu at mean $\mu_i$. The law of $Z_T$ is common across means, because exploration actions ignore all value feedback and the principal never observes that feedback. Let $\mathcal E_{i,T}$ be the event in Lemma~\ref{lem:app_iid_noise}. For $\textsf{G-UE}^{\mathrm{iid}}$, on this event each exploration reward mean differs from $v_i(a)$ by at most $\rho_{N_T}$. The selected action is therefore $2\rho_{N_T}$-optimal for the empirical menu.

For $\textsf{G-SE}^{\mathrm{iid}}$, consider an action $a$ that is active in a completed elimination round ending before $\tau$. Let $s$ be its number of Phase-2 selections through the end of that round. Its all-sample reward mean differs from its empirical-menu utility by
\[
\widehat r_{a,m_T+s}-v_i(a)
=\mu_i(\widehat x_{a,m_T+s}-\bar x_a)
 -(\widehat p_{a,m_T+s}-\bar p_a)
 +\frac{M_{i,a,m_T+s}}{m_T+s}.
\]
The first two terms come from changes in the empirical allocation and payment since Phase~1. Because both consistency tests have passed, the intervals in \eqref{eq:ue_se_trigger_test} overlap, so for $y\in\{x,p\}$,
\[
|\widehat y_{a,m_T+s}-\bar y_a|
\leq\rho_{m_T}+\rho_{m_T+s}\leq2\rho_{m_T}.
\]
Since $\mu_i\leq B$, the first two terms together have absolute value at most $2(B+1)\rho_{m_T}$. The last term averages the allocation-weighted value deviations and has absolute value at most $\rho_{m_T+s}$ on $\mathcal E_{i,T}$.

Each endpoint of the reward confidence interval is a further $\rho_{m_T+s}$ from the reward mean. Its distance from $v_i(a)$ is therefore at most
\[
2(B+1)\rho_{m_T}+2\rho_{m_T+s}\leq[2(B+1)+2]\rho_{m_T}.
\]
Every action selected after the first completed round and strictly before $\tau$ survived the preceding round. Since the endpoint bound holds through that round, Lemma~\ref{lem:app_replacement} bounds the action's empirical-menu utility gap by $2K$ times the endpoint bound. Thus, on $\mathcal E_{i,T}$, the counted Phase-2 selections are $\kappa_T$-optimal in the empirical menu, where
\[
\kappa_T=
\begin{cases}
2\rho_{N_T},&\text{for }\textsf{G-UE}^{\mathrm{iid}},\\
2K[2(B+1)+2]\rho_{m_T},&\text{for }\textsf{G-SE}^{\mathrm{iid}}.
\end{cases}
\]
The action counts $N_{i,a}$ again include only Phase-2 selections strictly before $\tau$, excluding the first elimination round for $\textsf{G-SE}^{\mathrm{iid}}$.

To adapt \eqref{eq:app_conditional_counts}, multiply each count by $\mathbf1_{\mathcal E_{i,T}}$ before taking its conditional expectation:
\begin{align}
c_{i,a}(z)
=\E_i[N_{i,a}\mathbf1_{\mathcal E_{i,T}}
       \mid Z_T=z],\qquad
J_i(z)=\sum_a c_{i,a}(z),\qquad
\Gamma_{i,T}(z)=J_i(z)/T .
\label{eq:app_iid_good_counts}
\end{align}
If $J_i(z)>0$, define the allocation--payment pair
$(X_{i,T}(z),P_{i,T}(z))$ by \eqref{eq:app_menu_rows}.
If $J_i(z)=0$, use the empirical-menu entry
$(\bar x_{a^*}(z),\bar p_{a^*}(z))$ for an action $a^*$ maximizing
$v_i(a)$, with fixed tie-breaking. In this case the pair has zero
accounting weight because $\Gamma_{i,T}(z)=0$.

The allocation--payment pair assigned to type $i$ is a mixture of menu entries that are $\kappa_T$-optimal at mean $\mu_i$. Since all pairs are formed from the same
empirical menu, which includes the opt-out action, the resulting static direct mechanism
$D_T$ is approximately IC and IR and belongs to $\mathcal F_{\kappa_T}$,
with $\mu_i$ replacing $\theta_i$.

The pathwise bounds \eqref{eq:app_gue_accounting} and \eqref{eq:app_gse_accounting} apply to the allocation and payment consistency tests with $C_{\mathrm{iid}}$ in place of $C_0$.
Apply them on $\mathcal E_{i,T}$ and discard its complement at a cost bounded by a constant times $\Pr_i(\mathcal E_{i,T}^c)$ in average units.
After integration over $Z_T$, they give \eqref{eq:app_menu_accounting} with
\[
e_T=
\begin{cases}
O\left(\dfrac{KN_T+1}{T}+\rho_{N_T}
             +\sqrt{\log T/T}+KT^{-3}\right),
 &\text{for }\textsf{G-UE}^{\mathrm{iid}},\\[6pt]
O\left(\dfrac{Km_T+2K+1}{T}
             +(1+Km_T/T)\rho_{m_T}+KT^{-3}\right),
 &\text{for }\textsf{G-SE}^{\mathrm{iid}}.
\end{cases}
\]

The construction above uses the value-noise bound to obtain the menu representation for the screening problem indexed by the means. Conditional on $\mu_i$, the current value is drawn independently of the action and mechanism, so the mean payoff identity \eqref{eq:iid_payoff_identity} gives $\bar u_{i,T}=\mu_i\bar x_{i,T}-\bar p_{i,T}$. With $U_i(D)=\mu_iX_i-P_i$, the allocation and payment bounds in \eqref{eq:app_menu_accounting} therefore give
\[
\left|\bar u_{i,T}-\E[\Gamma_{i,T}U_i(D_T)]\right|
\leq(\mu_i+1)e_T.
\]

The stationary regret rates in \eqref{eq:app_iid_stationary_rates} give consistency. The stationary cutoff argument in Lemma~\ref{lem:stationary_cutoff_outcomes} applies with expected trade utility $\mu_i-p+\delta_T$ at price $p-\delta_T$. With the utility bound above and the static IC and IR constraints evaluated at the means, the proof of Lemma~\ref{lem:app_menu_transfer} gives asymptotic individual rationality, convergence to the static revenue benchmark, and no dynamic rent extraction. This proves part 2 of Theorem~\ref{thm:unknown_iid_extension}.

\subsection{\texorpdfstring{Extraction from no-ER algorithms}{Extraction from no-ER algorithms} and exact rent guarantees}

\emph{Every no-ER rule permits extraction.} Fix the value laws $(F_i)_i$. Since $\Theta_\mu$ is finite, the selection argument in \eqref{eq:uniform_average_er_bound} gives a uniform cumulative regret upper bound $R_T^{\mathrm{cum}}=o(T)$ with means in place of types. The construction in Appendix~\ref{app:noer_extraction} therefore applies with $\theta_i=\mu_i$. The comparator differences in the test blocks use expected utilities, so the price margins are the same as in the known-type proof with $\mu_i$ in place of $\theta_i$. Couple each auxiliary run with the actual run at a fixed mean using the same private value draws and learner randomization. While mechanisms agree, actions and all observed feedback, including the values observed upon allocation, agree inductively. Thus, for each testing block $k$, the first auxiliary policy preserves the distribution of actions and observed outcomes through the end of that block.
When the estimated mean is correct, the second auxiliary run and the actual run agree throughout the continuation.
Equation~\eqref{eq:app_noer_revenue} follows with means in place of types. Choosing the identification error and price discount as in that proof gives the strategy in part 1 of Theorem~\ref{thm:unknown_iid_extension}.

\emph{Exact rent guarantees.} Fix an algorithm $\mathcal L$ satisfying no dynamic rent extraction, and hold $(F_i)_i$ fixed across priors. The conditional law of play at mean $\mu_i$ depends on the algorithm, the fixed value law $F_i$, and the principal's behavioral policy. Changing the prior while holding the algorithm, value laws, and the principal's policy fixed leaves the expected time-average payment $\bar p_{i,T}$ unchanged for every mean $\mu_i$. The same policy remains feasible. Thus the prior comparison in Proposition~\ref{prop:no_dre_revenue_ceiling} remains valid. The proof of that proposition therefore applies with $\mu_i$ in place of $\theta_i$. This establishes the same static revenue upper bound.

If $\mathcal L$ is also consistent, the proof of Lemma~\ref{lem:stationary_cutoff_outcomes} applies with $\mu_i$ in place of $\theta_i$. The principal posts $p^*(\pi)-\delta_T$ in every period, with $\delta_T$ chosen as in that proof. Revenue then converges to $\operatorname{Rev}^*(\pi)$, and the expected utility for each mean $\mu_i$ converges to the static screening rent $u^*(\mu_i,\pi)=(\mu_i-p^*(\pi))_+$. The revenue upper bound therefore gives $V_T(\pi,\mathcal L)\to\operatorname{Rev}^*(\pi)$, so this sequence of stationary strategies belongs to $\operatorname{AO}(\pi,\mathcal L)$. This gives $g_i(\pi,\mathcal L)\leq u^*(\mu_i,\pi)$, while no dynamic rent extraction supplies the reverse inequality. Finally, any consistent rule weakly improving the guarded rules' guarantees at all means and full-support priors itself satisfies no dynamic rent extraction, and therefore has these same exact guarantees. A strict improvement anywhere is impossible. With $(F_i)_i$ held fixed, each of $\textsf{G-UE}^{\mathrm{iid}}$ and $\textsf{G-SE}^{\mathrm{iid}}$ is therefore Pareto optimal among consistent algorithms with respect to its rent guarantees across type--prior pairs. This proves part 3 and completes the theorem.

\end{document}